\documentclass[aps,pra,10pt,twocolumn,superscriptaddress,floatfix,nofootinbib,showpacs,longbibliography]{revtex4-1}
\usepackage{makecell}
\usepackage{array}
\newcolumntype{C}[1]{>{\centering\arraybackslash}m{#1}}
\usepackage{booktabs,adjustbox}
\usepackage{float}
\usepackage{amsmath}
\usepackage[normalem]{ulem}
\usepackage[utf8]{inputenc}  
\usepackage[T1]{fontenc}     %Output what you want e.g., é, ł, a, ü
\usepackage[table]{xcolor}
\usepackage{color,soul}
\usepackage[british]{babel}  %Do hyphenation according to british english
\usepackage[sc,osf]{mathpazo}
\usepackage[colorlinks=true, citecolor=blue, urlcolor=blue]{hyperref}  
\usepackage{graphicx} % Package to insert exteral figures
\usepackage[babel]{microtype}  %Improves text justification
\usepackage{amsmath,amssymb,amsthm,bm,amsfonts,mathrsfs} %Usefull math packages

\usepackage{xspace}  %Useful to add space in macros
\usepackage{pgf,tikz}
\usepackage{xcolor}
\usepackage{multirow}
\usepackage{array}
\usepackage{bigstrut}
\usepackage{braket}
\usepackage{color}
\usepackage{natbib}
\usepackage{multirow}
\usepackage{mathtools}
\usepackage{float}
\usepackage{blkarray}
\usepackage[caption = false]{subfig}
\usepackage{xcolor,colortbl}
\usepackage{color}
\usepackage{rotating}
\usepackage{tikz}
\usepackage{tabularx}

\newcommand{\Tr}{\operatorname{Tr}}

\newcommand{\be}{\begin{equation}}
\newcommand{\ee}{\end{equation}}
\newcommand{\ba}{\begin{eqnarray}}
\newcommand{\ea}{\end{eqnarray}}
\newcommand{\ketbra}[2]{|#1\rangle \langle #2|}

\newtheorem{theorem}{Theorem}
\newtheorem{corollary}{Corollary}

\newtheorem{lemma}{Lemma}

\newcommand{\R}{\mathbb{R}}
\newcommand{\C}{\mathbb{C}}

\newcommand{\spc}[1]{\mathcal{#1}}

\newcommand{\Span}{{\mathsf{Span}}}

\def\>{\rangle}
\def\<{\langle}

\newcommand{\map}[1]{\mathcal{#1}}
\newcommand{\Proof}{{\bf Proof. \,}}

\usepackage{centernot}
\usepackage{subfig}

\begin{document}

\title{Quantum networks boosted by entanglement with a control system} 
\author{Tamal Guha}
\affiliation{QICI Quantum Information and Computation Initiative, Department of Computer Science, The University of Hong Kong, Pokfulam Road, Hong Kong}
\author{Saptarshi Roy}
\affiliation{QICI Quantum Information and Computation Initiative, Department of Computer Science, The University of Hong Kong, Pokfulam Road, Hong Kong}
\author{Giulio Chiribella}
\affiliation{QICI Quantum Information and Computation Initiative, Department of Computer Science, The University of Hong Kong, Pokfulam Road, Hong Kong}
\affiliation{Department of Computer Science, University of Oxford, Parks Road, Oxford OX1 3QD,  United Kingdom}
\affiliation{Perimeter Institute for Theoretical Physics, Caroline Street, Waterloo, Ontario N2L 2Y5, Canada}
\begin{abstract}
 %Quantum communication networks with coherent control over their configuration offer promising advantages in a number of   tasks. 
 Networks of quantum devices with coherent control over their configuration offer promising advantages in quantum information processing.   So far, the investigation of these advantages  assumed that the control system was  initially uncorrelated with the  data processed by the network.
    Here, we explore  the power of  quantum correlations  between data and control,   showing  two communication  tasks that can  be accomplished with information-erasing channels  if and only if the sender shares prior entanglement with a third party (the ``controller'') controlling the network configuration.     The first task is to  transmit classical messages without leaking information to the controller. 
     The second task is to establish   bipartite entanglement with a  receiver, or, more generally, to establish multipartite entanglement with a 
        number of spatially separated receivers.    
       % These tasks   reveal the potential of  correlated quantum control in quantum networks and can be experimentally demonstrated  with current photonic technologies.   
\end{abstract}
\maketitle
\textit{Introduction.}-A remarkable feature of quantum particles is the ability to  undergo  multiple  evolutions simultaneously, in a coherent quantum superposition \cite{aharonov1990superpositions, oi2003interference, aaberg2004operations, abbott2020communication, chiribella2019quantum, dong2019controlled, vanrietvelde2021universal}.   In a seminal work \cite{gisin2005error}, Gisin, Linden, Massar, and Popescu showed that the interference of multiple quantum evolutions could be used to filter out noise in quantum communication,   with potential benefits for quantum  key distribution and other quantum communication tasks.  Recently, the communication capacities arising from the  interference of multiple noisy channels have been studied in \cite{abbott2020communication, chiribella2019quantum, loizeau2020channel, kristjansson2021witnessing}.    Experimental demonstrations of the benefits of the interference of quantum evolutions have been provided in \cite{lamoureux2005experimental,rubino2021experimental}.

The superposition of quantum evolutions is generated by introducing a  control system, which determines the evolution undergone by a  target system.  Quantum networks equipped with control systems  provide a new  paradigm for quantum  information processing, and at the same time are an interesting toy model for investigating new  causal structures that could potentially arise  in a quantum theory of gravity   \cite{Hardy05, Hardy09, Hardy09a}. A concrete example here    is the quantum SWITCH  \cite{Chiribella09a, Chiribella13}, 
  a higher-order operation that connects two variable channels in an order determined by the state of a quantum system, giving rise  to a feature called causal non-separability  \cite{Oreshkov12,araujo2015witnessing}.    Over the past decade, the quantum SWITCH stimulated  several experimental investigations \cite{procopio2015experimental, Rubino17, Goswami18, wei2019experimental, Guo20} (see also \cite{Goswami20} for a review) and was found to offer information processing advantages in many tasks, including classification of quantum channels \cite{Chiribella12, araujo2014computational},   communication complexity \cite{guerin2016exponential},   quantum communication \cite{Ebler18,salek2018quantum,Chiribella21NJP, loizeau2020channel, caleffi2020quantum, procopio2020sending, Bhattacharya21, Sazim21, Chiribella21}, quantum metrology \cite{Zhao20,  chapeau2021}, and quantum thermodynamics \cite{Felce20, Guha20, Simonov21}.

Previous studies   on coherently controlled quantum networks explored  the benefits of quantum  superpositions of states of the control corresponding to definite  configurations.    In all these studies, the control was assumed to be  initially uncorrelated with the target.   It is possible, however, to consider a more general situation, in which the  control and the target share  prior correlations.  In this situation,  the data processed by the network becomes  correlated with its evolution,   potentially giving rise to new phenomena  that could not be observed   in the traditional setting. %\sout{In a similar spirit, we consider the situation where target-control entanglement can non-trivially make the path choice of the target indefinite,  generalizing the usual framework of superposition of path.} 

In this paper,  we explore the power  of quantum correlations between control and target,  showing that  they enable two communication tasks  that are impossible with an uncorrelated control, or even with a classically correlated one.    The tasks involve the assistance of a third party (the 	``controller'')  who has access to the control system and  shares initial quantum correlations with the sender. The role of the controller is to assist the receiver, by providing classical information gathered from the control.  For example, the controller could be a quantum communication company  responsible for  the connection between the sender and receiver. More generally, the controller could be any party who has access to the outcomes of measurements performed on the control.

Both tasks involve communication through noisy channels that  completely erase information when  used in a definite configuration. The first task is the communication of a classical message without leaking information to the controller.   We show that this task can  be perfectly achieved with  information-erasing channels if only if the sender and the controller initially share a maximally entangled state.    
% link privacy   to the availability of shared entanglement between the sender and the party holding the control system,   showing  that complete privacy can be attained \textit{if and only if} they are maximally entangled. 
The second task is  to establish  bipartite entanglement between a sender and a receiver,  or, more generally, to establish multipartite entanglement between the sender and a number of spatially separated receivers. 
In this case,  we show that  perfect entanglement can be established via information-erasing channels  if and only if the target and the control are initially in a maximally entangled state.

\medskip

%\textit{Quantum SWITCH with initial  correlations between control and target.}-   Let us start by summarizing the framework used in the rest of the paper.  A quantum process with input system $A$ and output system $B$ is mathematically described by a quantum channel, that is, a completely positive trace-preserving map, from the set of input density matrices $\mathcal{D}(\mathcal{H}_{A})$ over the Hilbert space $\mathcal{H}_{A}$ to the set of output density matrices  $\mathcal{D}(\mathcal{H}_{B})$. The action of a quantum channel $\map E$ on a generic density matrix $\rho \in D  (\spc H_A)$ can be conveniently expressed in the Kraus representation   $\map E(\rho)=\sum_{k}E_{k}\rho E_{k}^{\dagger}$, 
%where the Kraus operators $\{E_{k}\}$ satisfy the condition $\sum_{k}E_{k}^{\dagger}E_{k}=\mathbb{I}_{d_{A}}$,  $\mathbb I_{d_{A}}$  denoting the identity matrix in dimension $d_A$,  the dimension of the Hilbert space $\mathcal{H}_{A}$.
\textit{Quantum communication with entangled control.-}
 We  start by reviewing the mathematical description of coherent control over the configurations of quantum devices, focussing in particular on coherent control over the choice of quantum devices  and over their order.  For simplicity,  we  discuss  the case of $N=2$ channels, leaving the general case to the Supplemental Material.  
 
 The action of a quantum device  is mathematically described by a quantum channel, that is, a completely positive, trace-preserving linear map, acting on the density matrices of a given quantum system \cite{heinosaari2011mathematical}.   Quantum channels can be conveniently expressed in the Kraus representation $\map E  (\rho)   =   \sum_i  E_i  \rho  E_i^\dag$, where the Kraus operators $\{E_i\}$ satisfy the normalization condition $\sum_i  E_i^\dag E_i  =  I$, $I$ denoting the identity matrix on the system's Hilbert space.     Control over the order of two devices is described by the quantum SWITCH \cite{Chiribella09a, Chiribella13}, an operation that combines  two channels $\map E$ and $\map F$  acting on a target system,  generating a new channel $\mathcal{S}(\map E,\map F)$, acting jointly on the target and a control system. In the simplest version of the quantum SWITCH, the channel $\mathcal{S}(\map E,\map F)$ executes  the two channels $\map E$ and $\map F$ either in the order $\map E\circ \map F$  or in the order  $\map F\circ \map E$, depending on whether the control qubit is initialized in the state $|0\>$  or  $|1\>$, respectively. Explicitly, the control-order channel  $\mathcal{S}(\map E,\map F)$ is    specified by the relation
\begin{equation}\label{eq1}
\mathcal{S}(\map E,\map F)(\rho)=\sum_{i,j} S_{ij} \,\rho \, S_{ij}^{\dagger}
\end{equation}
with 
\begin{equation}\label{eq2}
    S_{ij}=F_{i} E_{j}\otimes\ketbra{0}{0}+E_{j} F_{i}\otimes\ketbra{1}{1}
\end{equation}
with $\{E_{j}\}$ and $\{F_{i}\}$ being the Kraus operators corresponding to the channels $\map E$ and $\map F$ respectively.    Note that the control-order channel $\mathcal{S}(\map E,\map F)$ depends only on the input channels $\map E$ and $\map F$, and not on the specific Kraus decompositions used in Eq. (\ref{eq2}).

Control over the choice of a noisy channel can be described in a similar way.  A quantum channel  $\map T$ that executes either  channel $\map E$ or  channel $\map F$  depending of the state of a control system has the form  \cite{aaberg2004operations, abbott2020communication, chiribella2019quantum} 
\begin{eqnarray}\label{eq3}
\map T (\rho) = \sum_{ij} T_{ij} \rho T_{ij}^\dagger
\end{eqnarray}
with
\begin{eqnarray}\label{eq4}
T_{ij} = E_i \, \beta_j \otimes |0\rangle\langle 0| + F_j\,  \alpha_i \otimes |1\rangle\langle1|,
\end{eqnarray}
where $\alpha_i$ and $\beta_j$ are complex amplitudes satisfying the normalization conditions $\sum_i  |\alpha_i|^2  =  \sum_j |\beta_j|^2   =1$.  

An important difference between  control over the  choice of two devices and  control over their order is that, while the control-order channel $\map S  (\map E,\map F)$ depends only on the channels $\map E$ and $\map F$, the control-choice channel $\map T$  depends  also on the amplitudes $\alpha_i$ and $\beta_j$ \cite{oi2003interference, aaberg2004operations, abbott2020communication, chiribella2019quantum, dong2019controlled, vanrietvelde2021universal}.  The physical reason for this dependence is that controlling the channel choice  means choosing  which channel is {\em not} used, or equivalently, which channel is  fed  a trivial input, such as {\em e.g.} the vacuum state \cite{chiribella2019quantum}.  Modelling the trivial input as a state $|{\rm triv}\>$ orthogonal to all states of the target system,  the choice-controlled channel $\map T$ can be regarded as a function of two {\em extended channels} $\widetilde {\map E}$ and $\widetilde {\map F}$ with Kraus operators $\widetilde E_i  =   E_i  +  \alpha_i  \,  |{\rm triv}\>\<{\rm triv}|$ and $\widetilde F_j  =   F_j  +  \beta_j  \,  |{\rm triv}\>\<{\rm triv}|$, respectively \cite{chiribella2019quantum, vanrietvelde2021universal}.   
%The channel $\map T$ depends only the channels $\widetilde {\map E}$ and $\widetilde {\map F}$, and not on the particular choice of Kraus operators used to represent them. 
For this reason, in the following we will use the notation $\map T  (\widetilde {\map E}  ,  \widetilde {\map F})$.   

\begin{figure}
    \centering
    \includegraphics[width = 0.9\linewidth]{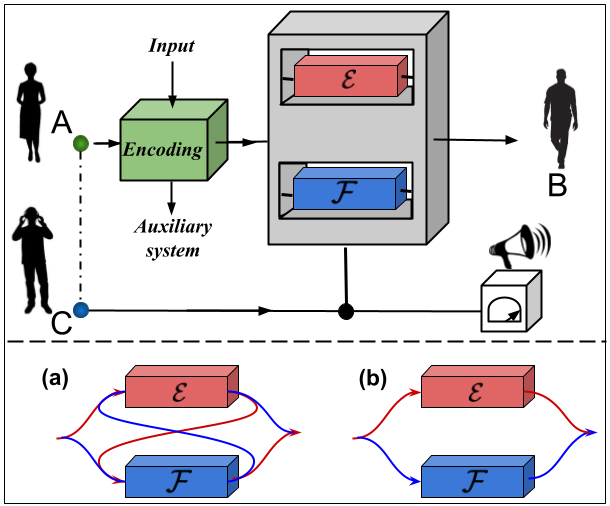}
    \caption{Quantum communication with the assistance of correlations with a control system.  Sender A communicates to receiver B through two noisy channels with the assistance of a third party,  C, who controls the configuration of the two channels.    We focus on the case where the configuration is either the order of the noisy channels (a), or the choice of which channel is used (b).  The controller and the sender initially share an entangled state (dotted line on the top left).  Then, the sender encodes some input data, by performing local operations  on her part of the entangled state. The output of these operations  is a signal that is sent through the network, and possibly some auxiliary systems that the sender will keep in her laboratory.  After transmission, the controller assists the receiver by providing  him classical information extracted from the control system.   %{\color{blue} [{Please remove ${\tt CTRL}  (\map E,\map F)$ and replace $\map D$ with something like $Encoding$. Also, it would be good to add $Input~data$ on the top, and $Auxiliary~systems$ on the bottom]} }
}    \label{fig1}
\end{figure}

   In all previous works  it has been assumed that the target and  control are initially uncorrelated, so that the effective evolution of the target can be interpreted as a superposition of evolutions corresponding to different  configurations of the network. 
   Here we will go beyond this assumption, exploring  the scenario where the control and the target are initially in a joint quantum state, with the net result that the evolution of the target is correlated with its state.  We  will consider a communication scenario where the target system travels from a sender (Alice) to a receiver (Bob), while the control system is held by a third party (Charlie), called the controller.     The role of the controller is to assist  Alice and Bob in their communication task.  In our protocols, Charlie's assistance will be limited to one round of classical communication to Bob.   
   
   The initial entanglement between target and control can then be  regarded as an offline resource, initially shared by Charlie and Alice, independently of the messages  she will send to Bob.      To encode information,  Alice will perform local operations on her side of the entangled state, which will later sent through the channel, as illustrated in Fig. \ref{fig1}.

To highlight the power of quantum correlations, we consider the extreme case where the channels $\map E$ and $\map F$ completely erase information, producing a fixed  pure state for every possible initial state of the target system.   These channels play a fundamental role in quantum thermodynamics, where they serve as the basis for extending   Landauer's principle to the quantum domain~\cite{rio2011thermodynamic} and  for evaluating the work cost of quantum processors~\cite{faist2015minimal}.   We will refer to these channels as {\em information-erasing channels}.  Taken in isolation, information-erasing channels have no ability to transmit any type of information, be it classical or quantum.  In the following we will focus on the case where $\map E$ and $\map F$ are  {\em orthogonal} information-erasing channels, that is, information-erasing channels that  output  orthogonal pure states, hereafter denoted as  $|0\>$ and $|1\>$, respectively.  In the case of control over the choice we consider the extended channels  $\widetilde{\map E}$ and $\widetilde  {\map F}$ with  Kraus operators $\widetilde E_0  =   |0\>\<0|  +   \,  |{\rm triv}\>\<{\rm triv}|$,  $\widetilde E_1  =   |0\>\<1|$,  $\widetilde F_0  =   |1\>\<0|$, and 
  $\widetilde F_1  =   |1\>\<1|  +   \,  |{\rm triv}\>\<{\rm triv}|$, respectively.   The benefit of this setting  is that the control-order  and  control-choice channels coincide, namely 
   \begin{align}\label{K}
  \map S  (\map E, \map F)  =  \map T    (\widetilde{\map E},  \widetilde{\map F})    =: \map K \, ,
  \end{align}
  as one can readily verify from the definitions.   
This observation allows us to treat the order and the choice in a unified way, without specifying which type of control we are considering.     It is worth stressing, however, that the identification in Eq. (\ref{K}) holds only for specific extensions $\widetilde {\map E}$ and $\widetilde{\map F}$, and that these  extensions are {\em not} information-erasing channels on the larger space spanned by the three states $|0\>,  |1\>,$ and $|{\rm triv}\>$.

   % when the two channels $\map E$ and $\map F$ are information-erasing channels outputting orthogonal states.    \par
%When two information-erasing channels  $\map E_\psi$ and $\map E_\phi$ are placed into the quantum SWITCH $\mathcal{S}(\mathcal{E}_{\psi},\mathcal{E}_{\phi})$,  the resulting channel has  
 %Kraus operators
   % \begin{equation}\label{eq4}
    %      K_{ij}=\<i|\phi\>   \ketbra{\psi}{j}\otimes\ketbra{0}{0}+\<j|\psi\>\,  \ketbra{\phi}{i}\otimes\ketbra{1}{1}
    %\end{equation}
  % as one can see from   Eq. (\ref{eq1}) by inserting the Kraus operators $E_i  =  |\phi\>\<i| $ and   $F_j =  |\psi\>\<j|$ for the two information-erasing channels.    The  generalization of the above expression for $d$ information-erasing channels placed in $d$ alternative orders is  provided  in the Supplementary Material.  In the following we will present three communication protocols that can be achieved through the switched channel $\map S (\map E_{\psi},\map F_{\phi})$.
   % when the two channels $\map E$ and $\map F$ are information-erasing channels outputting orthogonal states.    \par 

%\textcolor{blue}{\section{}}
\textit{Private classical communication}- A sender, Alice, wants to communicate a bit of classical information  to a distant receiver Bob.  She wants the communication to be secure, in the sense that no other party except Bob  can access the message.  Unfortunately, Alice and Bob do not share a secret key,  and therefore  protocols like the one-time pad are not viable. Still,  Alice has the assistance of a third party, Charlie, who controls the configuration of two communication channels,  as in Figure     \ref{fig1}.
Charlie can share entangled states with Alice, and can assist the communication by sending classical information to Bob.  However, Charlie should not be able to extract any information about Alice's message, otherwise the privacy requirement would be compromised.   

We now show that the desired task can be achieved perfectly using coherently-controlled information-erasing channels.  The crucial observation is that  the channel $\map K$ in Eq. (\ref{K})  has a decoherence-free  subspace \cite{Lidar98, Lidar2000, Lidar2003, Lidar2014} spanned by the states $|0\> \otimes |0\>$ and $|1\> \otimes |1\>$ (see the Supplemental Material for a detailed analysis).  This subspace contains Bell states $|\Phi^{\pm}\>  =  (  |0\> \otimes |0\> \pm |1\> \otimes |1\>)/\sqrt 2$, which can be generated from  $|\Phi^+\>$ by performing local unitary operations.   Hence, Alice can encode a bit $x\in \{+,-\}$ in one of the states $|\Phi^{\pm}\>$  and send it through the channel $\map  K $ without encountering any noise.  On the other hand, Charlie has no access to the value of the bit, because the states $|\Phi^{\pm}\>$ cannot be distinguished  using only measurements on the control system.  In the end,  Charlie  measures the control  on the  Fourier basis $\{|+\>  , |-\>\}$, with $|\pm\>  :  =  ( |0\> \pm  |1\> )/\sqrt 2$, and communicates the outcome to Bob, who also measures on the Fourier basis.     If Charlie's outcome is $+$,    then Bob's outcome is Alice's original bit. If Charlie's outcome is $-$, then  Bob only needs to flip the value of his bit, thus obtaining the value of Alice's bit.

In the Supplemental Material we show that   maximal  entanglement between target and control is strictly necessary: for  information-erasing channels $\map E$ and $\map F$,  Alice can perfectly communicate a bit in a way that is oblivious to Charlie  {\em only if} Alice and Charlie initially share a maximally entangled two-qubit state.   In addition, we provide an extension of the above results from qubits to general $d$-dimensional systems:   
 \begin{theorem}\label{theo1}
%Order-target entanglement enables the private  classical communication of a bit  through two orthogonal  information-erasing channels.
 A classical  $d$it can be communicated, with no leakage to the controller,  through $d$ orthogonal information-erasing channels in $d$ coherently controlled configurations if and only if the control and target are initially in a $d$-dimensional maximally entangled state.
\end{theorem}
   
\begin{figure}
    \centering
    \includegraphics[width = 0.8\linewidth]{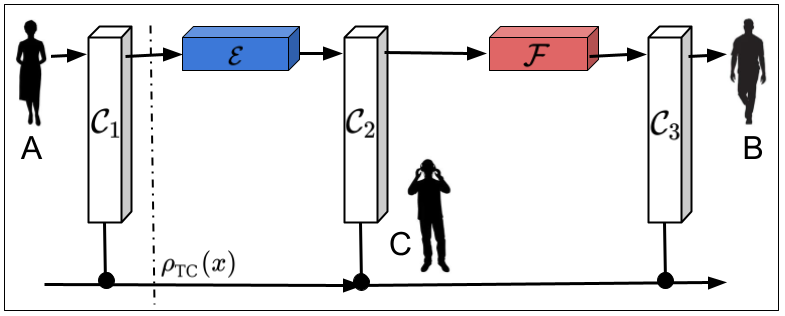}
    \caption{General protocol with  two noisy channels $\map E$ and $\map F$ in a fixed configuration and controlled operations before and after $\map E$ and $\map F$.   Protocols of this type can perfectly transmit classical messages from a sender to a receiver, but necessary leak information to the controller.}
    \label{fig2}
\end{figure}   
   
   Theorem \ref{theo1}  highlights the advantage of quantum correlations between the target and control systems.
   % In the Supplemental Material \cite{supple}, we also provide a quantitative relation between the amount of entanglement and the number of bits that Alice can reliably transmit to Bob without leaking information to Charlie.   
   Moreover, it also highlights a fundamental difference between protocols using control over the channels configurations, and protocols using the noisy channels $\map E $ and $\map F$ in a fixed configuration, while allowing control over  operations performed before and after each noisy channel  \cite{Gurin19}, as  illustrated in Figure \ref{fig2}.     
    These protocols  allow Alice to send classical information to Bob through the control, in a way that is completely independent of the  noisy channels $\map E$ and $\map F$  \cite{Kristjnsson20}.  However, this kind of protocols generally leak information to Charlie, violating the privacy requirement of our communication task.      When $\map E$ and $\map F$ are information-erasing channels, the leakage of information to Charlie is strictly necessary, as we  prove  in the Supplemental Material.

%\textcolor{blue}{\section{Sharing Entanglement}}
\textit{Establishing entanglement.}- %Our second protocol highlights the benefit of target-control correlations in a quantum communication task. The task is to  establish entanglement, which can later be used for transmitting quantum information.  
Our second task is to establish entanglement between the sender and  a receiver, or more generally, a number of spatially separated receivers.  Let us consider first the case of a single receiver, Bob.   
%Using Charlie's assistance, Alice and Bob can establish entanglement  accomplish this task through two orthogonal information-erasing channels $\map E$ and $\map F$.    
 Initially, Alice and Charlie share a maximally entangled state. Then, Alice converts it into the  GHZ state $(|0\>\otimes |0\>\otimes |0\>  + |1\>\otimes|1\> \otimes |1\>)/\sqrt 2$, by  applying a $\tt CNOT$ gate on the target qubit and on an additional reference qubit, present in her laboratory and initially in the state $|0\>$.  Alice keeps the reference qubit with her, and sends the target qubit through the controlled channel $\map K$.     The presence of the decoherence-free subspace $\Span\{  |0\>\otimes |0\>,  |1\>\otimes |1\> \}$ guarantees that the GHZ state is preserved by the channel.    At this point, Charlie measures the control qubit on the Fourier basis $\{|+\>,  |-\>\}$  and announces the result to Bob, who does nothing if the result is +, and performs a Pauli $Z$ correction if the outcome is -.  The net result of the protocol is that Alice and Bob share the maximally entangled state $|\Phi^+\>$,   which can later be used for quantum communication.
%At this point, Alice and Bob  can use the entangled state as a resource for quantum teleportation \cite{Bennett93},  using a classical communication line as a side channel.  The overall result of this protocol is that Alice can transfer an arbitrary qubit state to Bob with the assistance of Charlie.  Note that one could also consider a variant of the protocol where Alice encodes  a  qubit state $\alpha |0\>  +  \beta \,  |1\>  $ directly in the states $\alpha |0\>\otimes |0\>  +  \beta \,  |1\>\otimes |1\>$, where the first and second qubit become the target and the control, respectively.  However,  this approach would expose Alice and Bob to leak information to Charlie.  In the GHZ-based protocol, instead, Alice and Bob can first test the quality of the entanglement they obtained through the protocol, e.g. using an entanglement verification protocol  \cite{Acin12, Coladangelo17, Bowles18}.   In this way, they can assess the potential leakage {\em before} sending any information, and they can decide to proceed only if the potential leakage is below a desired threshold.  

The above protocol can be generalized to dimension $d$, using $d$ orthogonal information-erasing channels and quantum control over $d$ orders.   Also in this case, we show that maximal entanglement between target and control is strictly necessary:

\begin{theorem}\label{theo3}
Coherent control on the configuration  of $d$ orthogonal information-erasing channels 
%Composing $d$ information-erasing channels  in a superposition of $d$ orders
enables perfect establishment  of  a maximally entangled two-qudit state  if and only if the sender and controller initially share a $d$-dimensional maximally entangled state. 
\end{theorem}
See the Supplementary Material for the proof.

\begin{figure}
    \centering
    \includegraphics[width = 0.8\linewidth]{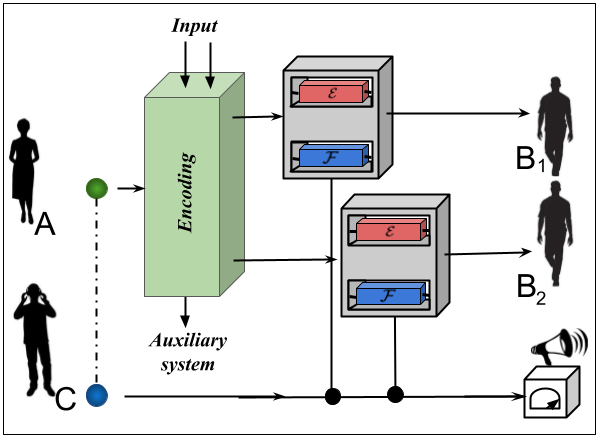}
    \caption{Distribution of entanglement to $N=2$  spatially separated parties  through coherently controlled information-erasing channels. %{\color{blue} [Please remove ${\tt CTRL}{  \map E,\map F}$ and make two separate control lines for the two controlled channels. ]}  
    The task can be perfectly achieved with the assistance of shared entanglement between the qubit at the sender's end, and a qubit used to control the configuration of channels between the sender and each receiver. }
    \label{fig3}
\end{figure}

We now extend the protocol to the case of  $N$ spatially separated receivers, each of which  is connected to the sender through   coherently-controlled  information-erasing channels, as  in Figure \ref{fig3}.  

 The generalization to $N>1$ receivers has two important features.  First, we will show that the dimension of the control system can be kept constant, independently of $N$.  In other words,  the amount of control required by the protocol is asymptotically negligible in the large $N$ limit.    The second feature is that our protocol transmits perfect  $(N+1)$-partite GHZ states, which can be used as a primitive in many applications, including communication complexity \cite{buhrman1999multiparty},  multiparty cryptography \cite{brassard2007anonymous}, secret sharing and entanglement verification \cite{broadbent2009ghz}, and quantum sensor networks \cite{komar2014quantum,eldredge2018optimal,qian2019heisenberg}.  In the context of quantum communication,  GHZ states can be used to achieve  a task known as random receiver quantum communication (RRQC)  \cite{Bhattacharya21}, where the goal is to transfer quantum information to one  of many receivers, whose identity is  disclosed only after the transmission phase.  Strikingly, entanglement with the control allows us to achieve RRQC with information-erasing channels, whereas in the lack of such entanglement RRQC can only be achieved with quantum channels that preserve classical information   \cite{Bhattacharya21}.   Once again, our results highlight the power of quantum correlations  between the sender and the controller.

Let us see how the protocol works. Initially, Alice and Charle share a two-qubit maximally entangled state.  Then, Alice converts it in to an  $(N+2)-$qubit GHZ-state, by applying CNOT gates on her qubit and $N$ additional qubits in her laboratory.  At this point, Alice sends $N$ out of the $(N+1)-$qubits from her part of the GHZ state to the $N$  receivers.   Crucially, the controlled channel  preserves the GHZ state (see the Supplemental Material). At this point, Charlie  performs a Fourier measurement  on his qubit and communicates the result to one of the $N$ Bobs, who performs a local correction operation, leaving the remaining $N+1$ qubits (one with the sender and $N$ of them with Bob) in the GHZ state.

Also in this case, we prove that entanglement between control and target is strictly necessary for a perfect distribution of GHZ states. This result and its $d$-dimensional generalization are contained in the following theorem  
\begin{theorem}
Coherent control on the configuration of $d$ orthogonal information-erasing channels enables perfect establishment of $d$-dimensional GHZ states between the sender and  $N$ spatially separated receivers if and only if the sender and controller initially share a $d$-dimensional maximally entangled state.
\end{theorem}

 \par%our protocol shows that the target-control entanglement allows one to achieve the desired communication task with noisier channels: here the task is achieved with constant channels, which have zero classical capacity, while in  \cite{Bhattacharya21} it was achieved with channels that had maximal classical capacity} %\textcolor{blue}{[@Giulio: You added this portion before the \textit{distribution} part and according to our last meeting we have transferred this part here.]} .Importantly, we also conclude that it is impossible to achieve the tasks reported in this letter by employing the framework of superposition of path, which does not require any indefiniteness in the causal order of quantum operations.

\textit{Conclusion.}-- In this work we initiated the exploration of quantum networks whose configuration is entangled with the state of a control system. We focussed on applications to quantum communication, identifying two tasks that can be perfectly achieved if and only if the sender and the controller initially share  maximal entanglement.

Our first task, the transmission of classical messages without leakage to the controller of the network's configuration, highlights a fundamental difference between  protocols where the configuration of the  channels is coherently controlled, and protocols where the configuration is fixed and controlled operations are allowed before and after each channel: when the channels completely erase information, no protocol that uses them in a fixed configuration can achieve   private communication between the sender and the receiver.   Our second  task highlights the benefits of sender-controller entanglement for establishing entanglement with one or more receivers.   
%In the qubit case, the tasks we proposed could be implemented experimentally using standard photonic techniques, such as the preparation of Bell states and GHZ states, together with the control techniques developed in \cite{lamoureux2005experimental,procopio2015experimental, Rubino17, Goswami18, wei2019experimental, Guo20,Goswami20,rubino2021experimental}.  
While in this work we focussed on quantum communication, we believe that quantum correlations with the configuration of quantum networks will have significant implications  also in other  quantum technology,   likely including  quantum metrology, thermodynamics, and  computation.

\section*{Acknowledgments}   We acknowledge  fruitful discussions with  H Kristjánsson.    This work is by the Hong Kong Research Grant Council through grant 17307719  and though the Senior Research Fellowship Scheme SRFS2021-7S02, by the Croucher Foundation,  and by the John Templeton Foundation through grant  61466, The Quantum Information Structure of Spacetime  (qiss.fr).  Research at the Perimeter Institute is supported by the Government of Canada through the Department of Innovation, Science and Economic Development Canada and by the Province of Ontario through the Ministry of Research, Innovation and Science. The opinions expressed in this publication are those of the authors and do not necessarily reflect the views of the John Templeton Foundation.
	\bibliographystyle{apsrev4-1}
	\bibliography{bib}
\begin{widetext}

\medskip

  \begin{center}
\large{\textbf{Supplementary Information}}%\par
%\Large{\textbf{Quantum networks boosted by entanglement with a control system}}
\end{center}
\section{Private classical communication}
\subsection{Proof  of Theorem 1,  \textit{if} part}
To prove the \textit{if} part for Theorem. 1, here we  show that the sender, Alice, can convey $\log d$ bit of classical information privately to the receiver Bob by encoding the classical information $x\in\{0,1,\dots,d-1\}$ with a local operation on one side of the shared maximally entangled state $\ket{\Phi^{+}}\in  \mathbb{C}^{d} \otimes \C^d$.
Consider $d$ orthogonal information-erasing channels $\{\mathcal{E}_{j}\}_{j=0}^{d-1}$ acting on the set of density matrices $\mathcal{D}(\mathcal{H}_{d})$ over a $d$-dimensional Hilbert space $\mathcal{H}_{d}$. A set of Kraus  for the channel $\mathcal{E}_{j}$ is
\begin{equation}\label{eq1}
    E_{i_j}^{(j)}=\ketbra{j}{i_j}, \qquad i_j\in\{0,\dots, d-1\} \, .
\end{equation}
  We now add quantum control over the order of the $d$ information-erasing channels, allowing a $d$-dimensional control system to  select one out of $d$ cyclic permutations.
  The resulting channel is  \cite{procopio2020sending,Sazim21,Chiribella21} 
\begin{equation}
    \mathcal{S}(\mathcal{E}_{0}, \mathcal{E}_{1}, \cdots, \mathcal{E}_{d-1})(\rho_{\text{AC}})=\sum_{i_{0},i_{1},\cdots,i_{d-1}}S_{i_{0},i_{1},\cdots,i_{d-1}}\rho_{\text{AC}}S_{i_{0},i_{1},\cdots,i_{d-1}}^{\dagger}
\end{equation}
with  Kraus operators
\begin{equation}\label{eq2}
    S_{i_{0},i_{1},\cdots,i_{d-1}}=   \sum_{j=0}^{d-1}   \,   E^{(j)}_{i_j}   E^{( j \oplus 1)}_{i_{j\oplus1}}   \cdots  E^{(j \oplus (d-1) )}_{i_{j\oplus  (d-1)}} \otimes |j\>\<j|  \, ,
    % E_{i_{0}}^{(0)}E_{i_{1}}^{(1)}\cdots E_{i_{d-1}}^{(d-1)}\otimes\ketbra{0}{0}+E_{i_{1}}^{(1)}\cdots E_{i_{d-1}}^{(d-1)}E_{i_{0}}^{(0)}\otimes\ketbra{1}{1}+\cdots+ E_{i_{d-1}}^{(d-1)}E_{i_{0}}^{(0)}\cdots E_{i_{d-2}}^{(d-2)}\otimes\ketbra{d-1}{d-1}
\end{equation}
where $\oplus$ denotes the sum modulo $d$.  Using Eq. (\ref{eq1}), we  rewrite Eq. (\ref{eq2}) in the following compact form
\begin{eqnarray}\label{eq3}
S_{i_{0},i_{1},\cdots,i_{d-1}}&=    \sum_{j=0}^{d-1}   \,   s_j~ |j\>\<i_{j  \ominus 1}|   \otimes |j\>\<j|  \, , \qquad \qquad s_j:    =  \prod_{l\not = j \ominus 1}   ~  \<  i_l  |  l\oplus 1\>    \, ,
%    \delta_{i_{0},1}\delta_{i_{1},2}\cdots\delta_{i_{d-2},(d-1)}\ketbra{0}{i_{d-1}}\otimes\ketbra{0}{0} + \delta_{i_{1},2}\cdots\delta_{i_{d-2},(d-1)}\delta_{i_{d-1},0}\ketbra{1}{i_{0}}\otimes\ketbra{1}{1}\nonumber\\
%&+\cdots\cdots\cdots +\delta_{i_{d-1},0}\delta_{i_{0},1}\cdots\delta_{i_{d-3},(d-2)}\ketbra{d-1}{i_{d-2}}\otimes\ketbra{d-1}{d-1}
\end{eqnarray}
where $\ominus$ denotes subtraction modulo $d$.  

At this point, there are three possible cases:
\begin{enumerate}
\item  $i_{l} =   l \oplus 1 $  for all  $l\in \{0,\dots, d-1\}$, 
\item   $i_{l} =   l \oplus  1 $  for all  $l$ except one,  or equivalently,  $i_{j\ominus 1}   =  j$  for all $j$ except one, 
\item  $i_{l} \not =   l \oplus  1  $  for two or more values of $l$.
\end{enumerate}
In Case 1, the Kraus operator is $S_{1,2,\dots, d-1, 0}  =   \sum_{j=0}^{d-1}  |j\>\<j| \otimes |j\>\<j|=:  P_0$.   
In Case 2, the Kraus operators are of the form  $|j\>\<i_{j\ominus 1} |\otimes |j\>\<j|$, where $j$ is the one index for which  $i_{j} \not =  j\ominus  1 $. In Case 3, the Kraus operator $S_{i_0,i_1, \dots,  i_{d-1}}$ is zero. 
Summarizing, the controlled-order channel is given by 
\begin{align}\label{orderclear}
    \mathcal{S}(\mathcal{E}_{0}, \mathcal{E}_{1}, \cdots, \mathcal{E}_{d-1})(\rho_{\text{AC}})   =  P_0  \rho_{\text{AC}}  P_0  +  \sum_{j=0}^{d-1}  \sum_{l\not  =  j}    \<  l|\< j|  \rho_{\text{AC}}  |l\>  |j\>  ~  |j\>\<j|  \otimes |j\>\<j|  \, . 
\end{align} 

The same channel is obtained from a controlled choice of the information-erasing channels $\{\map E_{j}\}_{j=0}^{d-1},$ provided that one adopts the extended channels $\{\widetilde{\map E_{j}}\}_{j=0}^{d-1}$ with Kraus operators
$\widetilde E^{(j)}_{i_j}   = |j\>\<i_j|  +    \<i_j| j\> \,  |{\rm triv}\>\<{\rm triv}|$.  Indeed, the controlled-choice channel is given by \cite{abbott2020communication, chiribella2019quantum} 
\begin{equation}\label{e6}
    \mathcal{T}(\widetilde{\mathcal{E}}_{0}, \widetilde{\mathcal{E}}_{1}, \cdots, \widetilde{\mathcal{E}}_{d-1})(\rho_{\text{AC}})=\sum_{i_{0},i_{1},\cdots,i_{d-1}}T_{i_{0},i_{1},\cdots,i_{d-1}}\rho_{\text{AC}}T_{i_{0},i_{1},\cdots,i_{d-1}}^{\dagger}
\end{equation}
with  Kraus operators
\begin{equation}\label{T}
    T_{i_{0},i_{1},\cdots,i_{d-1}}=\sum_{j=0}^{d-1}    t_j \,   |j\>\<i_j|   \otimes |j\>\<j|   \, ,
    \qquad \qquad   t_j:  =   \prod_{l\not  = j}   \alpha_{i_l}^{(l)}    \, ,
\end{equation}
where $\alpha^{(l)}_{i_l}$ are the amplitudes associated to the $l$-th channel.  If we set $\alpha^{(l)}_{i_l}  =  \<i_l|  l\>$, then there are three possible cases: 
\begin{enumerate}
\item  $i_l  =  l$ for every $l\in  \{0,\dots, d-1\}$, 
\item $i_l =  l $ for every $l$ except one,  
\item $i_l\not  =  l$ for two or more values of $l$. 
\end{enumerate}
In Case 1, the Kraus operator is $T_{0,1,\dots,  d-1}  =  P_0$.  In Case 2, the Kraus operator is $|l\>\<i_l| \otimes |l\>\<l|$, where $l$ is the one value such that $i_l \not =  l $. In Case 3, the Kraus operator $T_{i_0,i_1,\dots ,  i_{d-1}}$ is zero.  
 Summarizing, we obtained the relation  
\begin{align}\label{e8}
    \mathcal{S}(\mathcal{E}_{0}, \mathcal{E}_{1}, \cdots, \mathcal{E}_{d-1})   =    \mathcal{T}(\widetilde{\mathcal{E}}_{0}, \widetilde{\mathcal{E}}_{1}, \cdots, \widetilde{\mathcal{E}}_{d-1})  =:  \map K  \, ,
\end{align}
which proves Eq. (5) in the main text and generalizes it to $d\ge 2$.  In the following we will treat the controlled-order and controlled-choice in a unified way, referring to the channel $\map K$.

Note that the channel $\map K$ has a decoherence free subspace spanned by the vectors $|j\>  \otimes |j\>$,  $j\in  \{0,\dots, d-1\}$.  Hence, it preserves the maximally entangled states
$$\ket{\Phi_{x}}_{\rm AC}= \frac{1}{\sqrt{d}} \sum_{j=0}^{d-1} \, e^{\frac{2\pi  i\,  j x}{d}}  \ket{j}_{\rm A} \otimes \ket{j}_{\text{C}}, \qquad x\in\{0,1,\cdots,d-1\}.$$ 

Since  all these states are  maximally entangled, they are locally preparable from the canonical maximally entangled state $\ket{\Phi^+}_{\text{AC}}= \sum_{j=0}^{d-1} \, \ket{j}_{\rm A} \otimes \ket{j}_{\text{C}}/\sqrt d  $ by means of suitable local unitary operations on Alice's side.  %Precisely, one has $\ket{\Phi_{x}}_{\rm AC}=( U_{x,\rm A}\otimes I_{C}) \ket{\Phi^{+}}_{\rm AC},~\forall x\in\{1,2,\cdots, d-1 \}$, with  $ U_x: =\sum_j \,  \exp(\frac{2\pi  i  \,  j x}{d})  |j\>\<j|$. 
Therefore, Alice can encode  $\log d$ bits by locally transforming the pre-shared $\ket{\Phi^{+}}_{\text{AC}}$ one of these $d$ maximally entangled states. Then, she can send her part of the state to Bob through the controlled quantum channels.   After the transmission,  Bob and Charlie will  share one of these $d$ maximally entangled states. 

%Note that, by applying the set of unitary $\{ U^{(m)}\}_{m=0}^{d-1}$ on the state $|f_0\>=\frac{1}{\sqrt{d}}\sum_{k}\ket{k}\in\mathbb{C}^{d}$ one can prepare a orthonormal basis $|f_m\>:= U^{(m)}\ket{f_0}\}_{m=0}^{d-1}\in\mathbb{C}_{d}$. This, in turn, guarantees that the set of states $\{\ket{\Phi_{\text{BC}}^{(m)}}\}_{m=0}^{d-1}$ can be perfectly discriminated under one-way LOCC \cite{Bandyopadhyay2011}.  

The states $\{  |\Phi_x\>_{\rm BC}\}$ can be perfectly discriminated under one-way LOCC.  The protocol is simple:  Charlie and Bob  perform two independent measurements on the Fourier basis $\{|f_m\>= \sum_{j}   e^{\frac{2\pi i\,  jm}d}|j\>/\sqrt d \}_{m=0}^{d-1}$, and  Charlie  communicates his outcome to Bob.  The joint probability distribution of their outcomes $m_{\rm B}$ and $m_{\rm C}$ is $p(m_{\rm B},  m_{\rm C})  =  \delta_{m_{\rm B}  +  m_{\rm C}  ,  x}/d$, and allows Bob to infer the value of the message $x$ from his outcome $m_{\rm B}$ and from Charlie's  $m_{\rm C}$.     At the same time,  Charlie will remain completely blind about the transmitted message, as his measurement outcome alone contains no information about $x$.

\subsection{Proof  of Theorem 1,  \textit{only if} part}
In the previous section we have shown that Alice can communicate $\log d$ bits of classical information privately to Bob via $d$ controlled pin-maps, provided that she initially shares a   $d$-dimensional maximally entangled state    with the controller Charlie. We now prove that   maximally entangled states are strictly necessary for this communication task. Precisely, we will show that a perfect  communication of $\log d$ classical bits through coherently controlled information-erasing channels is possible only if Alice and Charlie  initially share a bipartite state $\rho^*_{\rm AC}$ that can be locally converted into the  $d$-dimensional maximally entangled state  $|\Phi^+\>_{\rm AC}$. \par

The proof is rather complex and makes use of a series of lemmas, proved in the following.   All throughout this section, we will use the following notation:   $\rho_{\rm AC}^*$ will be the state shared by Alice and Charlie at the beginning of the protocol,  $\map A_x$ be the local operation used by Alice to encode message $x$,    $\rho_{x, {\rm AC}}: =  (\map A_x\otimes \map I_{\rm C})  (  \rho_{\rm AC}^*)$ will be the joint state of Alice's and Charlie's systems right before transmission through the controlled channel, and   $\rho_{x,{\rm BC}}'$ will be the state of Bob's and Charlie's systems right after transmission. 

%\begin{lemma}
%Non-maximally entangled states between the sender and controller do not allow  communication of $\log d$ classical bits through coherently controlled information-erasing channels.  
%\label{lemma:2}
%\end{lemma}
\begin{lemma}\label{lem:maxentoutput}
Perfect communication of $\log d$ bits through coherently controlled information-erasing channels is possible only if  the final states $\{\rho_{x,{\rm BC}}'\}_{x=0}^{d-1}$ are pure, orthogonal, and maximally entangled. 
\end{lemma}

\Proof Let  $\map C$ be either the controlled-order  channel     $\mathcal{S}(\mathcal{E}_{0}, \mathcal{E}_{1}, \cdots, \mathcal{E}_{d-1})$ 
   or the controlled-choice channel    $\mathcal{T}(\widetilde{\mathcal{E}}_{0}, \widetilde{\mathcal{E}}_{1}, \cdots, \widetilde{\mathcal{E}}_{d-1})$, so that $\rho_{x,{\rm BC}}'  =  \map C   (\rho_{x,{\rm AC}})$.    
   
   Now,  $\map C$ transforms every density matrix into a density matrix with support contained in the subspace $\spc H_0  :  = \Span \left(  |j\>\otimes |j\>\, , \,  j\in  \{0, 1,\dots, d-1\}  \right)$.    This fact can be readily checked from Eq. (\ref{eq3}) and (\ref{T}) in the cases of controlled order and controlled choice, respectively.  
 
  Since the subspace  $\spc H_0$ is $d$-dimensional, the only way to achieve the perfect communication of $\log d$ bits is that the   states   $\{\rho_{x,{\rm BC}}'\}$ are pure and orthogonal, say $\rho_{x,{\rm BC}}'  =  |\Phi_x\>\<\Phi_x|_{\rm BC}$ where $\{ |\Phi_x\>_{\rm BC}\}_{x=0}^{d-1}$ is an orthonormal basis for the subspace $\spc H_0$.

 We now show that each state $|\Phi_x\>$ must be maximally entangled.  By definition, we have 
 \begin{align}\label{aaa}
|\Phi_x\>\<\Phi_x|_{\rm BC}    =  \map C    \Big(   (\map A_x\otimes \map I_{\rm C}) \, (\rho_{{\rm AC}}^*    \Big)  \, .
\end{align}
 Let us write  $|\Phi_x\>_{\rm BC}  =  \sum_j\,  c_{x,j}  \,  |j\>_{\rm B}\otimes |j\>_{\rm C}$.     Multiplying both sides of Eq. (\ref{aaa}) by $I_{\rm B}  \otimes |j\>\<j|_{\rm C}$ on the left and on the right, we obtain 
      \begin{align}
\nonumber       |c_{x,j}  |^2    \,  |j\>\<j|_{\rm B} \otimes |j\>\<j|_{\rm C} & =  (I_{\rm B}  \otimes |j\>\<j|_{\rm C}  ) \,    \map C  
  \Big(      (\map A_x\otimes \map I_{\rm C} \, (\rho_{{\rm AC}}^*)    \Big)    \, (I_{\rm B}  \otimes |j\>\<j|_{\rm C} )  \\
\nonumber   & =     \map C  
  \Big(    (I_{\rm A}  \otimes |j\>\<j|_{\rm C}  ) \,      (\map A_x\otimes \map I_{\rm C} \, (\rho_{{\rm AC}}^*)   \,  (I_{\rm A}  \otimes |j\>\<j|_{\rm C}  )  \Big)   \\
  \nonumber   & =     \map C  
  \Big(       \map A_x  (\sigma_{j,{\rm A} })    \otimes |j\>\<j|_{\rm C}       \Big)  \qquad \qquad \sigma_{j,{\rm A}}:  =   (  I_{\rm A} \otimes  \<  j|_{\rm C}  )  \,  \rho_{\rm AC}^* \, (  I_{\rm A} \otimes  | j\>_{\rm C\>}  )  \\
   \nonumber  &  =     \map E_j   
        \map A_x    (  \sigma_{j,{\rm A}})  \otimes |j\>\<j|     \\
       &  =  p_j ~ |j\>\<j|_{\rm B}\otimes |j\>\<j|_{\rm C}    \,,   \qquad  \qquad p_j:  =   \Tr [ \sigma_{j,\rm A}] \, ,      
    \end{align}
 where the second equality follows from the expression of the Kraus operators of $\map K$ (Eqs. (\ref{eq3}) and (\ref{T}) for the controlled-order and controlled-choice, respectively),  the forth equality follows from the fact that $\map K$ is a controlled information-erasing channel, and the fifth equation follows from the fact that $\map A_x$ is trace-preserving.

Since $j$ and $x$ are arbitrary, we conclude that $|c_{x,j}|^2  =  p_j$ for every $x$ and $j$.  Now, recall that the vectors $\{|\Phi_x\>\}$ form an orthonormal basis for the subspace $\spc H_0$, and therefore  $\sum_{x=0}^{d-1}  |\Phi_x\>\<\Phi_x|  = \sum_{j  =  0}^{d-1} |j\>\<j| \otimes |j\>\<j|$.  Multiplying both sides of this equation  by $\<j|_{\rm C}$ on the left and $|j\>_{\rm C}$ on the right, we obtain  
\begin{align}
\sum_{x=0}^{d-1}   |c_{x,j}|^2   |j\>\<j|     =  |j\>\<j| \, ,  
\end{align}
which combined with the fact that $|c_{x,j}|^2$ is independent of $x$,  implies $|c_{x,j}|^2  =    1/d$ for every $j$.  In conclusion, the states $|\Phi_x\>$ are maximally entangled.  \qed

\medskip

 To continue the proof, we consider separately the cases of the controlled order and the controlled choice.  

\medskip

{\bf Proof for controlled order.}     The proof uses the following lemma:

 \begin{lemma}\label{lem:maxentorder}
Perfect communication of $\log d$ bits with  $d$ information-erasing channels in a controlled order is possible only if the initial state $\rho_{\rm AC}^*$ is locally convertible into a $d$-dimensional maximally entangled state.  
 \end{lemma}
  
\Proof   By Lemma \ref{lem:maxentoutput}, perfect communication is possible only if the states $\rho'_{x, \rm BC}$ are pure,  orthogonal, and maximally entangled.  Then, one has  $ \rho_{ x,{\rm BC}}'  = |\Phi_x\>\<\Phi_x|_{ {\rm BC}}  = \map S  (\map E_0, \map E_1,\dots, \map E_{d-1})  \,  (\rho_{x, {\rm AC}})$.  

A necessary condition for  the state $\map S  (\map E_0, \map E_1,\dots, \map E_{d-1})  \,  (\rho_{x, {\rm AC}})$ to be maximally entangled is that the separable terms in Eq. (\ref{orderclear}) vanish, or equivalently, that  $ |\Phi_x\>\<\Phi_x|_{ {\rm BC}}    =    P_0  \,  \rho_{x, {\rm AC}}  P_0$.  Since $P_0$ is a projector (up to the inessential relabelling of the first space as $\rm B$ or $\rm A$), the normalization of the state $P_0  \,  \rho_{x, {\rm AC}}  P_0$ implies that $P_0  \,  \rho_{x, {\rm AC}}  P_0  =   \rho_{x, {\rm AC}}$, and therefore,    $\rho_{x, {\rm AC}}  =  |\Phi_x\>\<\Phi_x|_{ {\rm AC}}$.    In summary, all the states $\{\rho_{x, {\rm AC} }\}_{x=0}^{d-1}$ are maximally entangled.   But these states are obtained by performing local operations on the initial state $\rho_{\rm AC}^*$.  Since local operations cannot increase entanglement, we conclude that $\rho_{\rm AC}^*$ must be  locally convertible into a $d$-dimensional maximally entangled state. \qed

 \medskip  

Combining Lemmas \ref{lem:maxentoutput} and \ref{lem:maxentorder}, we obtain the desired necessity proof for the controlled-order of information-erasing channels:  perfect communication of $\log d$ bits is possible only if the initial state shared by Alice and Charlie is (locally equivalent to) a $d$-dimensional maximally entangled state.  
 
 \medskip

{\bf Proof for controlled choice.}    The proof   is more subtle than the proof for controlled order, because there are infinitely many possible  "controlled-choice channels," depending on which extensions $\widetilde {\map E}_j$ are used.      Our proof will hold for all possible choices. 

To get started, we need a general fact on the controlled-choice of $d$ information-erasing channels:
\begin{lemma}\label{lem:Tdecomposition}
The controlled-choice channel  $\mathcal{T}(\widetilde{\mathcal{E}}_{0}, \widetilde{\mathcal{E}}_{1}, \cdots, \widetilde{\mathcal{E}}_{d-1})$  can be written as   
   \begin{align}\label{Tdecomposition}
      \mathcal{T}(\widetilde{\mathcal{E}}_{0}, \widetilde{\mathcal{E}}_{1}, \cdots, \widetilde{\mathcal{E}}_{d-1})(\rho_{\text{AC}})   =  T_{0,\dots,  0}  \, \rho_{\rm AC} \,  T_{0,\dots,  0}^\dag   +  \sum_{j=0}^{d-1} \sum_{i_j\not =  0}  ~  \Tr \left[ \left(I  -  |v_j \>\<  v_j|\right)_{\rm A} \otimes |j\>\<j|_{\rm C}  ~   \rho_{\rm AC}  \right]  ~ ~ |j\>\<j|_{\rm B} \otimes |j\>\<j|_{\rm C} \, ,
    \end{align}
    where $\{ |v_j\>\}$ are suitable vectors satisfying $\|  |v_j\> \| \le 1$ for every $j\in \{0,\dots,  d-1\}$,  and   
    \begin{align}\label{e13}
    T_{0,\dots, 0 }   =   \sum_{j=0}^{d-1}  |j\>_{\rm B}\<v_j|_{\rm A} \otimes |j\>_{\rm C}\<j|_{\rm C} \, .
    \end{align}
\end{lemma}

\Proof    The proof uses a property of extended channels proven in \cite{vanrietvelde2021universal}: for every extended channel $\widetilde {\map E}$ there exists a Kraus representation with operators of the form $\widetilde E_i  =  E_i  +   \alpha_i  ~  |{\rm triv}\>\<{\rm triv}|$  such that $\alpha_0 = 1$ and $\alpha_i  =  0 $ for every $i>0$.  Applying this result to the channels $\widetilde {\map E}_j$, we obtain  Kraus representations 
\begin{align}
\widetilde E^{(j)}_{i_j}  =  |j\>\<  v_{i_j}^{(j)} | +        \alpha_{i_j}^{(j)}  ~  |{\rm triv}\>\<{\rm triv}| \, ,
\end{align}  
where $\left(   |v^{(j)}_{i_j}\>\right)_{i_j}$ are (possibly non orthonormal) vectors satisfying the  normalization condition $\sum_{i_j}  |v_{i_j}^{(j)}\>\<v_{i_j}^{(j)}|=  I$ for every $j$.   In this representation, the controlled-choice channel reads \cite{abbott2020communication, chiribella2019quantum} 
\begin{equation}\label{calT}
    \mathcal{T}(\widetilde{\mathcal{E}}_{0}, \widetilde{\mathcal{E}}_{1}, \cdots, \widetilde{\mathcal{E}}_{d-1})(\rho_{\text{AC}})=\sum_{i_{0},i_{1},\cdots,i_{d-1}}T_{i_{0},i_{1},\cdots,i_{d-1}}\rho_{\text{AC}}T_{i_{0},i_{1},\cdots,i_{d-1}}^{\dagger}
\end{equation}
with  Kraus operators
\begin{equation}\label{Tall}
    T_{i_{0},i_{1},\cdots,i_{d-1}}=\sum_{j=0}^{d-1}    t_j \,   |j\>\<v_{i_j}^{(j)}|   \otimes |j\>\<j|   \, ,
    \qquad \qquad   t_j:  =   \prod_{l\not  = j}   \alpha_{i_l}^{(l)}    \, . 
\end{equation} 
At this point, there are three possible cases:
\begin{enumerate}
\item  $i_{j} =   0  $  for all  $j$, 
\item   $i_{j} =  0 $  for all  $j$ except one,  
\item  $i_{j} \not = 0$  for two or more values of $j$.
\end{enumerate}
In Case 1, the Kraus operator is $T_{0,\dots, 0}  =   \sum_{j=0}^{d-1}  |j\>\<v^{(j)}_0| \otimes |j\>\<j|$.   
In Case 2, the Kraus operators are of the form  $|j\>\<  v_{i_j}^{(j)}  |\otimes |j\>\<j|$, where $j$ is the one index for which  $i_{j} \not =  0 $. In Case 3, the Kraus operator $T_{i_0,i_1, \dots,  i_{d-1}}$ is zero.  Inserting these expressions into Eq. (\ref{calT}), we  obtain  
 \begin{align}\label{choiceclear}
 \nonumber \mathcal{T}(\widetilde{\mathcal{E}}_{0}, \widetilde{\mathcal{E}}_{1}, \cdots, \widetilde{\mathcal{E}}_{d-1})(\rho_{\text{AC}})  &=
   T_{0,\dots, 0}  \rho_{\rm AC}  T_{0, \dots, 0}^\dag   +  \sum_{j=0}^{d-1}  \, \sum_{i_j\not  =  0}  \<  v_{i_j}^{(j)} |_{\rm A}  \<j|_{\rm C}    \rho_{\rm AC}  |  v_{i_j}^{(j)}\>_{\rm A} |j\>_{\rm C} ~ |j\>\<j|_{\rm B} \otimes |j\>\<j|_{\rm C} \\
     &  =  
   T_{0,\dots, 0}  \rho_{\rm AC}  T_{0, \dots, 0}^\dag   +  \sum_{j=0}^{d-1}  \,  \Tr [   (I_A  -  |v_0^{(j)}\>\<v_0^{(j)}|  )  \otimes |j\>\<j|  \,  \rho_{\rm AC}]    ~ |j\>\<j|_{\rm B} \otimes |j\>\<j|_{\rm C}  \, ,   \end{align}
   the second equation following from the  normalization condition $\sum_{i_j}  |v_{i_j}^{(j)}\>\<v_{i_j}^{(j)}|=  I$ for every $j$.
Defining $|v_j\>:  =  |v_0^{(j)}\>$ we then obtain Eq. (\ref{Tdecomposition}). \qed  
 
\medskip 

We now use the previous lemma to characterize the structure of the input states that give rise to orthogonal states in the output.  To this purpose, recall  Lemma \ref{lem:maxentoutput}, which states  that the states $\{\rho'_{x, \rm BC}\}$ are orthogonal only if they are  maximally entangled.  
 
\begin{lemma}\label{lem:subspace}
If  the state $\rho_{x,{\rm BC}}'$ is maximally entangled, then $\|    |v_j\>   \|  =1$ for every $j \in  \{0,\dots, d-1\}$, and the state $\rho_{x,{\rm AC}}$ has support contained in the subspace spanned by the vectors $ \{  |v_j\> \otimes |j\>\}_{j=0}^{d-1}$. 
\end{lemma}   

\Proof   Recall that $\rho_{x,{\rm BC}}'   =   \mathcal{T}(\widetilde{\mathcal{E}}_{0}, \widetilde{\mathcal{E}}_{1}, \cdots, \widetilde{\mathcal{E}}_{d-1})  (\rho_{x,{\rm AC}})$. For this state  to be maximally entangled, the separable terms in Eq.   (\ref{Tdecomposition}) must vanish. 
These terms vanish if and only if 
\begin{align}
 \Tr \left[ \left(I_{\rm A}  -  |v_j \>\<  v_j  |\right) \otimes |j\>\<j|  ~   \rho_{x,{\rm AC}}  \right] =  0   \qquad \forall j\in  \{0,\dots, d-1\} \, .
\end{align}   
This  condition implies the relation $\rho_{x,{\rm AC}}~  (|v_j\>\<v_j |_{\rm A}\otimes |j\>\<j|_{\rm C})=  0$ for every $j$ such that $\|  v_j\>  \|  <  1$.
% where the vectors $|v_{i_j}^{(j)}\>$ are defined  in the proof of Lemma \ref{lem:Tdecomposition}.    
In turn, this condition implies that the output state in Eq. (\ref{calT}) becomes
\begin{align}\label{bbb}
\mathcal{T}(\widetilde{\mathcal{E}}_{0}, \widetilde{\mathcal{E}}_{1}, \cdots, \widetilde{\mathcal{E}}_{d-1})  \,  (\rho_{x,{\rm AC}})
  =   T_{0,\dots,  0}  \, \rho_{x,{\rm AC}} \,  T_{0,\dots,  0}^\dag  =  T_*  \, \rho_{x,{\rm AC}} \,  T_*^\dag \, ,\end{align}
with
\begin{align}\label{e20}
T_*  =  \sum_{j \in S_* }  ~  |j\>\<v_j| \otimes |j\>\<j| \,,
\end{align}
$S_*$ being the set of values of $j$ such that $  \| \,  |v_j\>\,  \|=1$.   

The normalization of the states in Eq. (\ref{bbb}) implies that the state $\rho_{x,{\rm AC}} $ has support contained in the vector space spanned by  the vectors  $ \{  |v_j\> \otimes |j\>\}_{j\in  S_*}$.   Moreover, the condition that the state  $T_*  \,\rho_{j,{\rm AC}}  T^\dag_*$  be maximally entangled implies that   the set $S_*$ must contain all values of $j$.   Hence, the condition $  \| \,  |v_j \>\,  \|=1$ must be satisfied for every $j\in  \{0,\dots, d-1\}$.  \qed

We now show that  the states sent by Alice and Charlie through the channel must be maximally entangled.  

\medskip  

\begin{lemma}\label{lem:maxent}
If the states $\{\rho_{x,{\rm BC}}'\}_{x=0}^{d-1}$  are  orthogonal and maximally entangled,  then  the states $\{\rho_{x,{\rm AC}}\}_{x=0}^{d-1}$ are maximally entangled.  
%   and the vectors $ \{  |v_0^{(j)}\>\}_{j=0}^{d-1}$ are orthonormal.  
\end{lemma}   

\Proof 
  Since the states $\{\rho_{x,{\rm BC}}'\}_{x=0}^{d-1}$ are obtained from the states  $\{\rho_{x,{\rm AC}}\}_{x=0}^{d-1}$ through the action of a quantum channel, the former are orthogonal only if the latter are orthogonal.  

By Lemma \ref{lem:subspace}, the support of the states  $\{\rho_{x,{\rm AC}}\}_{x=0}^{d-1}$ is contained in the $d$-dimensional subspace spanned by the vectors $ \{  |v_j\> \otimes |j\>\}_{j=0}^{d-1}$.    Since the states $\{\rho_{x,{\rm AC}}\}_{x=0}^{d-1}$ are $d$ orthogonal states in a $d$-dimensional subspace,   they must be pure.    Let us write them as $\rho_{x{\rm AC}}  =  |\Psi_x\>\<\Psi_x|_{\rm AC}$, with 
\begin{align}
|\Psi_x\>_{\rm AC}  =  \sum_{j=0}^{d-1}  \,   \lambda_{x, j}  \,  |v_j\>_{\rm A} \otimes |j\>_{\rm C} \, .
\end{align} 

We now show that the orthogonal states $\{  |\Psi_x\>_{\rm AC}\}_{x=0}^{d-1}$ must be maximally entangled.  
First, recall that one has  $|\Psi_x\>\<\Psi_x|_{\rm AC} =  (\map A_x \otimes \map I_{\rm C})  (\rho_{\rm AC}^*)$. Tracing out both sides on the equation with $I_{\rm A} \otimes |j\>\<j|_{\rm C} $ we obtain 
\begin{align}
|\lambda_{x, j}|^2   &  =  \Tr[  (I_{\rm A} \otimes |j\>\<j|_{\rm C})  \,  (\map A_x \otimes \map I_{\rm C})  (\rho_{\rm AC}^*)    ]  =  \Tr[  (I_{\rm A} \otimes |j\>\<j|_{\rm C})  \, \rho_{{\rm AC}}^*     ]    =:p_j\, .
\end{align}
In short, $|\lambda_{x,j}|$ is independent of $x$.  

Moreover, since  the states $\{|\Psi_x\>_{\rm AC}\}_{x=0}^{d-1}$ are orthogonal, that is, they are a basis for the subspace spanned by the vectors $\{  |v_j\>_{\rm A}\otimes |j\>_{\rm C}\}_{j=0}^{d-1}$.   Hence, we have  
\begin{align}\sum_{x = 0}^{d-1}  |\Psi_x\>\<\Psi_x|_{\rm AC} =  \sum_{j=0}^{d-1}  \,  |v_j\>\<  v_j|_{\rm A} \otimes |j\>\<j|_{\rm C}  \, . \end{align}
Multiplying both sides of the equation by $\<j|_{\rm C}$ on the left and $|j\>_{\rm C}$ on the right, we obtain  
\begin{align}
\sum_{x=0}^{d-1}\,  |c_{xj}|^2  |v_j\>\<  v_j|_{\rm A}    =   |v_j\>\<  v_j|_{\rm A} \, ,
\end{align}
which implies $|c_{xj}|^{2}  = 1/d$  (recall that $\|  |v_j\> \|  = 1$ for every $j$ and therefore $|v_j\>$ cannot be the zero vector).  

Hence, the state $|\Psi_x\>_{\rm AC}$ can be rewritten as 
\begin{align}\label{ccc}
|\Psi_x\>_{\rm AC}  =  \frac 1{\sqrt d}\,  \sum_{j=0}^{d-1}  \,   e^{i \theta_{x, j}}  \,  |v_j\>_{\rm A} \otimes |j\>_{\rm C} \, ,
\end{align} 
for some suitable phases $\theta_{x,j} \in  \R$.     

To conclude that the vectors $|\Psi_x\>$ are maximally entangled, we show that the vectors $\{|v_j\>\}_{j=0}^{d-1}$ are mutually orthogonal.  
To this purpose, recall that all the states  $|\Psi_x\>$ must have the same marginal on system $\rm C$. The condition of equal marginals is 
\begin{align}
\sum_{j,l}  e^{i  (\theta_{x,j}  -  \theta_{x,l})}   \<  v_l|  v_j\>   \,  |j\>\<l|   =    \sum_{j,l}  e^{i  (\theta_{y,j}  -  \theta_{y,l})}    \<  v_l|  v_j  \>    \,  |j\>\<l|  \qquad \forall x,y \in  \{0,\dots, d-1\} \, .
\end{align}  
The equality holds if and only if $ e^{i  (\theta_{x,j}  -  \theta_{x,l})}   =   e^{i  (\theta_{y,j}  -  \theta_{y,l})}$  for every pair $(j,l)$ such that $ \<  v_l|  v_j\>   \not =  0$.  

On the other hand, no such pair can exist. The proof is by contradiction: suppose that there existed a pair $(j_1,j_2)$ such that $ \<  v_{j_1}|  v_{j_2}\>   \not =  0$.  Hence, there would exist a constant $\omega$ such that
\begin{align}
e^{i\theta_{x,j_2}}   =  \omega \,  e^{i\theta_{x,j_1}}  \qquad \forall x\in  \{  0,\dots, d-1\} \, .
\end{align}
This condition would imply that two columns of the matrix $M=  (  e^{i\theta_{x,j}})$ are proportional to each other, and therefore $\det (M)  =  0$.  But this would be  in contradiction with the fact that the states $\{|\Psi_x\>\}_{x=0}^{d-1}$ be orthogonal, which implies that the matrix $M$ has full rank.  

Hence, the condition  $ \<  v_l|  v_j\>  =  0$ must hold for every $j$ and $l$.  This implies that the vectors $\{   |v_j\>\}$ form an orthonormal basis, and therefore the states $\{|\Psi_x\>_{\rm AC} \}$ are maximally entangled.  \qed  

\medskip  

Putting everything together, we obtain  the desired result:  
\begin{lemma}\label{lem:maxentchoice}
Perfect communication of $\log d$ bits with a controlled choice of $d$ information-erasing channels is possible only if the initial state $\rho_{\rm AC}^*$ is locally convertible into a $d$-dimensional maximally entangled state.  
 \end{lemma}
  
\Proof By Lemma \ref{lem:maxentoutput}, perfect communication is possible only if the states $\rho'_{x, \rm BC}$ are orthogonal and maximally entangled.  Then, Lemma \ref{lem:maxent} implies that the states $\rho_{x,{\rm AC}}$ must be maximally entangled. Since these states are obtained  from the state $\rho_{\rm AC}^*$ by applying local operations,   the state $\rho_{\rm AC}^*$ must be maximally entangled.  \qed
  
 \medskip 
 
 Together, Lemmas \ref{lem:maxentorder} and \ref{lem:maxentchoice} conclude the proof of the {\em ``Only if''} part of Theorem 1 in the main text.

\subsection{Proof that  protocols with fixed configurations of the channels $\map E$ and $\map F$ cannot achieve private communication}    

Let $x$ be the bit value encoded by Alice, and let   $\rho_{{\rm TC }}  (x)$ be the joint state  of the target  and control after the first controlled operation in Figure 2 of the main text.    With the action of  the information-erasing channel $\map E$, the  target system is erased and reset it to the fixed state $|0\>$, leaving Charlie's system in the marginal state $\rho_{\rm C}  (x)  :  =  \Tr_{\rm T} [ \, \rho_{\rm TC} (x) \, ]$.    Now, the state of all systems at later times of the protocol depends only on the states $\rho_{\rm C}  (x)$.  For Bob to retrieve Alice's message, the states $\rho_{\rm C}  (0)$ and $\rho_{\rm C}  (1)$ must be perfectly distinguishable. But if they are perfectly distinguishable, they can be copied by Charlie, who can read Alice's message without being discovered.

\section{Establishing Entanglement with one receiver}
Here we will consider the scenario to establish $\log d$ ebits between Alice to Bob, through $d$ orthogonal information-erasing channels and a perfect side channel of quantum capacity $\log d$. Importantly, Charlie who has access on the side channel is only allowed to communicate classically with the receiver Bob. This prevents Alice to bypass the zero capacity channels via the perfect side channel. 
\subsection{Proof of Theorem 2., \textit{if} part}
Similar to the previous protocol, let us consider that Alice shares a maximally entangled state $\ket{\Phi^{+}}_{\text{AC}}\in\mathbb{C}_{d}^{\otimes2}$ with Charlie, beforehand. Now to establish the maximal entanglement with Bob, she will first prepare a $d$-dimensional quantum state $\ket{0}_{\text{A'}}$ and apply a joint unitary $\mathbb{V}_{\text{AA'}}$ on the two qudits she has at her possession. The action of the joint unitary is $\mathbb{V}_{\text{AA'}}\ket{k_{A}0_{A'}}=\ket{k_{A}k_{A'}},~\forall k\in\{0,1,\cdots,(d-1)\}$, which can also be identified as the perfect cloning machine for the orthonormal basis $\{\ket{k}\}_{k=0}^{d-1}$. Hence, the final tripartite state among the two qudits at Alice's lab and a single qudit at Charlie's lab will be genuinely entangled state, given by
%As the channel considered are same as that in the previous section, we can follow the same mathematics in this part. 
\begin{equation*}
    \ket{\psi}_{\text{AA'C}}=\frac{1}{\sqrt{d}}\sum_{j=0}^{d-1}\ket{jjj}_{\text{AA'C}}.
\end{equation*}
Now keeping the part $A$ with her, Alice (and Charlie) will send the qudit $A'$ (and $B$) through the controlled quantum channels of $d$ orthogonal information-erasing channels. The joint channel action hence can be depicted as
\begin{equation}\label{e28}
    I_{\text{A}}\otimes\mathcal{S}(\mathcal{E}_{0}, \mathcal{E}_{1}, \cdots, \mathcal{E}_{d-1})(\rho_{\text{AA'C}})=I_{\text{A}}\otimes\mathcal{T}(\mathcal{E}_{0}, \mathcal{E}_{1}, \cdots, \mathcal{E}_{d-1})(\rho_{\text{AA'C}}):=\sum_{i_{0},i_{1},\cdots,i_{d-1}}\widetilde{\mathcal{K}}_{i_{0},i_{1},\cdots,i_{d-1}}\rho_{\text{AA'C}}\widetilde{\mathcal{K}}_{i_{0},i_{1},\cdots,i_{d-1}}^{\dagger}
\end{equation}
where, $\widetilde{\mathcal{K}}_{i_{0},i_{1},\cdots,i_{d-1}}= I_{\text{A}}\otimes \mathcal{K}_{i_{0},i_{1},\cdots,i_{d-1}}$ and $\mathcal{K}_{i_{0},i_{1},\cdots,i_{d-1}}$ is same as in Eq.(\ref{e8}). 
This, in turn, assures that the controlled operation Eq. (\ref{e28}) maps any arbitrary three-qudit state to the  subspace spanned by $\ket{\psi}\otimes\ket{j}\otimes\ket{j}, \forall\psi \text{ and }\{j=0,1,\cdots,d-1\}$.
This directly follows from the structure of  the Kraus operators in Eq. \eqref{e28}, where there is no action on party A $(I_{\text{A}})$ and the operation on A$^\prime$C part $(\mathcal{K}_{i_{0},i_{1},\cdots,i_{d-1}})$ has a decoherence free subspace spanned by $\{\ket{j}\otimes\ket{j}\}, \text{ with } \{j=0,1,\cdots,d-1\}$.  
%This, in turn, assures that the controlled operation Eq. (\ref{e28}) has a decoherence free subspace spanned by $\ket{\psi}\otimes\ket{j}\otimes\ket{j}, \forall\psi \text{ and }\{j=0,1,\cdots,d-1\}$. 
Also observing that the $A'C$ marginal for the state $\ket{\psi_{\text{AA'C}}}$, i.e., $\rho_{\text{A'C}}=\text{Tr}_{\text{A}}(\ketbra{\psi}{\psi}_{\text{ATC}})$, is diagonal in the basis $\ket{j}\otimes\ket{j}, \{j=0,1,\cdots,d-1\}$ we conclude %Therefore, following the same argument as before, $\tilde{K}_{1,2,\cdots,(d-1),0}$ is the only Kraus operator can act on $\ket{\psi_{\text{AA'C}}}$ and hence similar to that of Eq. (\ref{eq5}), we find
\begin{eqnarray}\label{e29}
I_{\text{A}}\otimes\mathcal{S}(\mathcal{E}_{0}, \mathcal{E}_{1}, \cdots, \mathcal{E}_{d-1})\ketbra{\psi}{\psi}_{\text{AA'C}} &=& I_{\text{A}}\otimes\mathcal{T}(\mathcal{E}_{0}, \mathcal{E}_{1}, \cdots, \mathcal{E}_{d-1})\ketbra{\psi}{\psi}_{\text{AA'C}}\nonumber \\
&=& \ketbra{\psi}{\psi}_{\text{ABC}}
\end{eqnarray}
Hence, at the end Alice, Bob and Charlie share the same genuine entangled state among them. Now, considering the $d$-dimensional Fourier basis $\{\ket{f_{m}}=\frac{1}{\sqrt{d}}\sum_{j=0}^{d-1}\exp(i\frac{2\pi j m}{d})\ket{j}\}_{m=0}^{d-1}$ in Charlie's side, we can write 
\begin{equation*}
\ket{\psi_{\text{ABC}}}=\frac{1}{\sqrt{d}}\sum_{m=0}^{d-1}\ket{\Phi^{(m)}}_{\text{AB}}\otimes\ket{f_{m}}_{\text{C}},\text{ where, }\ket{\Phi^{(m)}}_{\text{AB}}=\frac{1}{\sqrt{d}}\sum_{j=0}^{d-1}\exp(i\frac{2\pi j m}{d})\ket{jj}_{\text{AB}}.
\end{equation*}
Evidently, Charlie, who has access on the control system can perform a measurement in $\{\ket{f_{m}}\}_{m=0}^{d-1}$ basis in his possession and communicate the outcome classically to Bob, who then apply a suitable unitary $ U_{m}$ on his qudit to get the state $\ket{\Phi^{+}}=\ket{\Phi^{0}}=\frac{1}{\sqrt{d}}\sum_{j=0}^{d-1}\ket{jj}$ between Alice and himself.\\
%The task we consider now is to share a bipartite entangled state between two distant parties, say Alice and Bob. Specifically, Alice prepares  some entangled state in her lab and intends to send it to Bob. However, since the quantum channel between Alice and Bob is noisy the sharing is imperfect and in some cases it might not be possible at all (for channels with zero quantum capacity). We demonstrate how the quantum switch with entangled target and control can achieve such a task with unit efficiency which cannot be realized without target control entanglement. Let us demonstrate it with a concrete example.
%Consider the d-dimensional information-erasing channels as defined in Eq. 1. The channel, composed of $d$ such maps, as mentioned before possesses zero classical and quantum capacity.  Alice prepares a maximally entangled state in her lab $\ket{\psi_{\text{AB}}}=\frac{1}{\sqrt{d}}\sum_{j=0}^{d-1}\ket{jj}_{\text{AB}} \in \mathbb{C}^{d}\otimes\mathbb{C}^{d}$ along with the control qubit in the state $|+_C\rangle = \frac{1}{\sqrt{d}}\sum_{i=0}^{d-1}|i\rangle_C$. She sends the qubit $B$ through the noisy channel whose causal order is made indefinite by the quantum switch. The final state shared between Alice, Bob and Charlie is computed as
%\begin{eqnarray}
% I_A \otimes \mathcal{S}_{BC}(\Lambda_0, \Lambda_1, \ldots , \Lambda_{d-1}) (\ket{\psi_{\text{AB}}} \otimes |+_C\rangle)
%\end{eqnarray}
\subsection{Proof of Theorem 2., \textit{only if} part}
With the help of the following Lemmas we will conclude that a pre-shared maximally entangled state, between Alice and Charlie, is necessary to establish $\log d$ ebit between Alice and Bob, using $d$ orthogonal qudit pin-maps.\par
Let us first consider the following result for the state shared between Alice, Bob and Charlie after the controlled quantum operation.

\medskip

\begin{lemma}\label{l7}
The state shared between Alice, Bob and Charlie after the controlled quantum operation must be three-qudit genuinely entangled GHZ state.
\end{lemma}
\begin{proof}
Let us consider the tripartite state produced after controlled quantum operation is $\rho_{\text{ABC}}$.

Now performing a measurement on his quantum system, Charlie will communicate the result to Bob. Depending upon which Bob will apply a local operation on his qudit to share a two qudit maximally entangled state among Alice and himself. 

Noting the fact that local operation and classical communication (LOCC) can not increase entanglement, the state $\rho_{\text{ABC}}$ should be maximally entangled in the $A|BC$ bipartition. Therefore the marginal of $A$ and at least one of $B$ or, $C$ should be $\frac{\mathbb{I}}{d}$.\par
Also from Eq. (\ref{orderclear}) and (\ref{choiceclear}), the two-qudit marginal of the state $\rho_{\text{ABC}}$ should be 
in the $d$-dimensional subspace spanned by $\ket{j}\otimes\ket{j}, ~j\in\{0,1,\cdots,d-1\}$. This two conditions together imply,
$$\sigma_{\text{A}}:=\text{Tr}_{\text{BC}}[\rho_{\text{ABC}}]=\frac{1}{d}\sum_{j=0}^{d-1}\ketbra{\psi_{j}}{\psi_{j}},~\text{and } ~ \sigma_{\text{BC}}:=\text{Tr}_{\text{A}}[\rho_{\text{ABC}}]=\frac{1}{d}\sum_{j=0}^{d-1}\ketbra{jj}{jj},$$
where the states $\{\ket{\psi_{j}}\}_{j=0}^{d-1}$ are orthogonal to each other.\par
Now, the condition that the state $\rho_{\text{ABC}}$ contains $\log d$ ebit in $A|BC$ bipartition, implies that the state is pure and can be written in the Schimdt form,
\begin{equation}\label{e30}
    \ket{\psi}_{\text{ABC}}=\frac{1}{\sqrt{d}}\sum_{j=0}^{d-1}e^{i \theta_{j}}\ket{\psi_{j}}_{\text{A}}\otimes\ket{jj}_{\text{BC}}.
\end{equation}
%where, $\{\ket{\phi_{j}}\}_{j=0}^{d-1}$ are set of orthogonal states lying in the subspace spanned by $\ket{j}\otimes\ket{j},~j\in\{0,1,\cdots,d-1\}$. 
This completes the proof.
\end{proof}

\medskip
Now, the above lemma further helps us to conclude a corollary regarding the state just before the controlled quantum operation.\par
After sharing an arbitrary two-qudit state $\rho_{\text{AC}}$ with Charlie, Alice prepares an ancillary system $\sigma_{\text{A'}}$ and apply any possible quantum operation $\Lambda_{\text{AA'}}$, which gives 
$$\rho^{*}_{\text{AA'C}}:=(\Lambda_{\text{AA'}}\otimes I_{\text{C}})\rho_{\text{AC}}\otimes\sigma_{\text{A'}}$$.

%We will now show that the state $\rho_{\text{AC}}$, shared between Alice and Charlie, necessarily be maximally entangled for both the controlled-order and controlled-choice configuration of $d$ orthogonal information-erasing channels.

%\medskip

\begin{corollary}\label{coro1}
The state $\rho^{*}_{\text{AA'C}}$ must be maximally entangled in $A|A'C$ bipartition.
\end{corollary}
\begin{proof}
After preparing, Alice sends the $A'C$ subsystems of state 
$\rho^{*}_{\text{AA'C}}$ through the controlled configuration of $d$ orthogonal information-erasing channels, which maps $A'C\mapsto BC$ and produces the state $\ket{\psi}_{\text{ABC}}$ (as in Eq. (\ref{e30})).\par
Since, LOCC on any bipartition of a multipartite state can not increase the entanglement, the state $\rho^{*}_{\text{AA'C}}$ should be maximally entangled in $A|A'C$ bipartition.  
\end{proof}

\medskip

Now, with the help of Lemma \ref{l7} and Corollary \ref{coro1}, we will finally conclude that regarding the necessity of sharing maximal entanglement between Alice and Charlie, separately for the controlled-order and controlled-choice configuration.

\medskip

\textbf{Proof for the controlled-order.}\par
\begin{lemma}\label{l8}
To establish $\log d$-bit entanglement between Alice and Bob after order-controlled configuration of $d$ orthogonal pin-maps, the state shared between Alice and Charlie must be maximally entanglement.
\end{lemma}
\begin{proof}
To preserve the maximal entanglement in $A|A'C$ bipartition of the state $\rho^{*}_{\text{AA'C}}$ under the controlled-order operation $I\otimes\mathcal{S}(\mathcal{E}_{0}, \mathcal{E}_{1}, \cdots, \mathcal{E}_{d-1})$, the separable terms in Eq. (\ref{e28}), i.e., in Eq. (\ref{orderclear}) should vanish. This implies,
$$\sum_{l\neq j}(I_{\text{A}}\otimes\ketbra{j}{l}_{\text{A'}}\otimes\ketbra{j}{j}_{\text{C}})\rho^{*}_{\text{AA'C}}(I_{\text{A}}\otimes\ketbra{l}{j}_{\text{A'}}\otimes\ketbra{j}{j}_{\text{C}})=0,~ ~ \forall j\in \{0,1, \cdots, d-1\}.$$
Therefore $\sigma^{*}_{\text{A'C}}:=\text{Tr}_{\text{A}}(\rho^{*}_{\text{AA'C}})$ will be orthogonal to the subspace spanned by $\ket{l}\otimes\ket{j},~\forall l,j\in\{0, 1, \cdots, (d-1)\}$ and $l\neq j$.\par
This, along with Corollary \ref{coro1} implies: $\sigma^{*}_{\text{A'C}}=\frac{1}{d}\sum_{j=0}^{d-1}\ketbra{jj}{jj}$
and hence the state $\rho^{*}_{\text{AA'C}}$ is pure, which can be written as
$$\ket{\psi^{*}}_{\text{AA'C}}=\frac{1}{\sqrt{d}}\sum_{j=0}^{d-1}e^{i\phi_{j}}\ket{\psi_{j}}_{\text{A}}\otimes\ket{jj}_{\text{A'C}},~\text{where } \<\psi_{k}\ket{\psi_{l}}=0, ~\forall k\neq l.$$
Note that, just by performing a measurement in $d$-dimensional Fourier basis 0f $\{\ket{\psi_{j}}\}$ on the subsystem $A'$, Alice can prepare a maximally entangled state between Charlie and herself. This, in turn, demands that the state $\rho_{\text{AC}}$, initially shared between Alice and Charlie, should be maximally entangled, otherwise Alice can increase entanglement only performing local operations in her lab.
\end{proof}

\medskip

\textbf{Proof for the controlled-choice.}\par
Let us first consider the controlled-choice configuration of $d$ orthogonal information-erasing channels acting on the three-qudit state $\rho^{*}_{\text{AA'C}}$. Following from Eq. (\ref{Tdecomposition}) we can write,
\begin{eqnarray}\label{e31}
\nonumber &\widetilde{\mathcal{T}}(\widetilde{\mathcal{E}}_{0},\widetilde{\mathcal{E}}_{1}, \cdots, \widetilde{\mathcal{E}}_{d-1})(\rho_{\text{AA'C}}):=I_{\text{A}}\otimes\mathcal{T}(\widetilde{\mathcal{E}}_{0},\widetilde{\mathcal{E}}_{1}, \cdots, \widetilde{\mathcal{E}}_{d-1})(\rho_{\text{AA'C}})\\&=I_{\text{A}}\otimes T_{0,0,\cdots,0}~\rho_{\text{AA'C}}~I_{\text{A}}\otimes T_{0,0,\cdots,0}^{\dagger}+\sum_{j=0}^{d-1}\text{Tr}_{\text{A'C}}[I_{\text{A}}\otimes(I_{\text{A'}}-\ketbra{v_j}{v_j})\otimes\ketbra{j}{j}_{\text{C}}~\rho_{\text{AA'C}}]\ketbra{j}{j}_{\text{B}}\otimes\ketbra{j}{j}_{\text{C}}
\end{eqnarray}
where, $||v_{j}||\leq1,~\forall j\in\{0, 1, \cdots, d-1\}$ and $T_{0, 0, \cdots, 0}$ is same as in Eq. (\ref{e13}).\par
Keeping this in mind we will now present the our main result in the following.
\begin{lemma}\label{l9}
It is possible to obtain the state $\ket{\psi}_{\text{ABC}}$ (as in Eq. (\ref{e30})) under controlled-choice configuration of $d$ orthogonal pin-maps, only if Alice and Charlie shares a maximally entangled state.
\end{lemma}
\begin{proof}
Following from Corollary \ref{coro1}, to preserve the maximal entanglement in the $A|A'C$ bipartition of the state $\rho^{*}_{\text{AA'C}}$ the separable terms in Eq. (\ref{e31}) must vanish. Therefore,
$$[I_{\text{A}}\otimes(I_{\text{A'}}-\ketbra{v_j}{v_j})\otimes\ketbra{j}{j}_{\text{C}}]\rho^{*}_{\text{AA'C}}=0,~\forall j\in\{0, 1, \cdots, d-1\}.$$
This further implies, $(I_{\text{A}}\otimes\ketbra{v_j}{v_j}_{\text{A'}}\otimes\ketbra{j}{j}_{\text{C}})\rho^{*}_{\text{AA'C}}=0, \forall ||v_{j}||< 1$, and hence we can rewrite Eq. (\ref{e31}) as
$$\widetilde{\map T}(\widetilde{\mathcal{E}}_{0},\widetilde{\mathcal{E}}_{1}, \cdots, \widetilde{\mathcal{E}}_{d-1})\rho^{*}_{\text{AA'C}}=(I_{\text{A}}\otimes T_{0,0,\cdots,0})~\rho^{*}_{\text{AA'C}}~(I_{\text{A}}\otimes T_{0,0,\cdots,0}^{\dagger})$$
where, $T_{0,0,\cdots,0}$ is same as in Eq. (\ref{e20}). Also identifying the form of $\widetilde{\map T}(\widetilde{\mathcal{E}}_{0},\widetilde{\mathcal{E}}_{1}, \cdots, \widetilde{\mathcal{E}}_{d-1})\rho^{*}_{\text{AA'C}}$ with that of Eq. (\ref{e30}), we can conclude that the $A'C$ marginal of the state $\rho^{*}_{\text{AA'C}}$ has support contained in the supspace spanned by $\ket{v_{j}}\otimes\ket{j}$ and $T_{0,0,\cdots,0}$ contains every $||v_{j}||=1,~\forall j\in\{0, 1,\cdots, d-1\}$.\par
Therefore, following from Corollary \ref{coro1}, the $A'C$ subsystem of the state $\rho^{*}_{\text{AA'C}}$ should be maximally mixed in the $d$-dimensional subspace spanned by $\ket{v_{j}}\otimes\ket{j}$, i.e., $\sigma^{*}_{\text{A'C}}=\text{Tr}_{\text{A}}(\rho^{*}_{\text{AA'C}})=\frac{1}{d}\sum_{j=0}^{d-1}\ketbra{v_j j}{v_j j}$ and also maximally entangled in the $A|A'C$ bipartition. Hence, the state $\rho^{*}_{\text{AA'C}}$ can be identified as a pure state given by,
\begin{equation}\label{e32}
\ket{\phi^{*}}_{\text{AA'C}}=\frac{1}{\sqrt{d}}\sum_{j=0}^{d-1}e^{i\theta_{j}}\ket{\psi_{j}}_{\text{A}}\otimes\ket{v_j}_{\text{A'}}\otimes\ket{j}_{\text{C}},
\end{equation}
where $\{\ket{\psi_{j}}\}_{j=0}^{d-1}$ is any orthonormal basis for the subsystem $A$. Therefore, applying a joint unitary $U_{\text{AA'}}$ Alice can transform the state $$\ket{\phi^{*}}_{\text{AA'C}}\to\ket{\xi^{*}}_{\text{AA'C}}:=\frac{1}{\sqrt{d}}\sum_{j=0}^{d-1}e^{i\theta_{j}}\ket{\psi_{j}}_{\text{A}}\otimes\ket{w_j}_{\text{A'}}\otimes\ket{j}_{\text{C}},$$
where $\{\ket{w_{j}}\}$ is an orthonormal basis for the subsystem $A'$, irrespective of the orthogonality condition for $\{\ket{v_j}\}$.\par
Now, performing a measurement on the subsystem $A$ of the state $\ket{\xi^{*}}_{\text{AA'C}}$ Alice can establish maximal  entanglement between Charlie and herself. This further demands that the state $\rho_{\text{AC}}$, initially shared between Alice and Charlie, should be maximally entangled. Otherwise, Alice will be able to increase entanglement only performing local operations at her possession.\par
This completes the proof.
\end{proof}
\section{Establishing multipartite entanglement with multiple receivers}
This section generalizes the results of the previous one from a single receiver to multiple, spatially separated receivers.  
   In this case, the task is to establish a $d$-dimensional GHZ-state  $\frac{1}{\sqrt{d}}\sum_{j=0}^{d-1}\ket{j}^{\otimes (N+1)}$ between Alice and $N$ Bobs.   This state represents a natural generalization of the canonical bipartite Bell state, and has applications in many quantum information processing tasks \cite{buhrman1999multiparty, brassard2007anonymous, broadbent2009ghz,komar2014quantum,eldredge2018optimal,qian2019heisenberg}.  
    Moreover, the GHZ state is impotant in that it maximizes   a distance-based measure of multipartite entanglement called the generalized geometric measure (GGM) \cite{ggm1,ggm2, ggm3, ggm4}, which for pure states admits the analytical expression
    \begin{eqnarray}
\label{Eq:GGM-final-express}
GGM(|\psi_{n}\rangle) = 1 - \max \big \lbrace \lambda_{A:B} | A\cup B = \lbrace 1,2,\ldots, n \rbrace, A\cap B = \emptyset \big \rbrace, \end{eqnarray}
where $\lambda_{A:B}$ denotes the maximal Schmidt number in the $A : B$ bipartition, and  the maximization is carried out over all possible bipartitions (note that  for mixed states  the computation of the GGM is generally hard  \cite{ggmmixed})   In the case of the $d$-dimensional GHZ state, the GGM assumes the maximum value  $\frac{d-1}{d}$.

\subsection{The Kraus Operators}
Let us first consider the Kraus operators for controlled order of $N$ noisy transmission lines for $N$ spatially separated Bobs, each consisting of $d$ orthogonal qudit information-erasing channels $\{\mathcal{E}_{0},\mathcal{E}_{1},\cdots,\mathcal{E}_{d-1}\}$, along with an identity channel on Alice's qudit.
%Identifying $\tilde{\mathcal{E}}_{k}=\mathcal{E}_{k}^{\otimes N}$ as the product of $N$ numbers of $\mathcal{E}_{k}$ information-erasing channel, we can write the controlled quantum operation on $(N+1)$ qudits as

\begin{equation}
    \mathbb{I}_{\text{A}}\otimes\mathcal{S}^{(N)}(\mathcal{E}_{0}^{\otimes N},\cdots,\mathcal{E}_{d-1}^{\otimes N})[\rho_{A A_1A_2\cdots A_NC}]=\sum_{I_{0},I_{1},\cdots,I_{d-1}}\widetilde{\mathcal{K}}_{I_{0}I_{1}\cdots I_{d-1}}[\rho_{A_1A_2\cdots A_NC}]\widetilde{\mathcal{K}}_{I_{1}I_{2}\cdots I_{d-1}}^{\dagger}
\end{equation}
where, $I_{k}$ is a $N$-tuple consisting of the following set of numbers $\{i_{k,n}\}_{n=1}^N$. %$(i_{0,n},i_{1,n},\cdots,i_{d-1,n})\in\{0,1,\cdots,(d-1)\}$. 
The individual Kraus operators can then be expressed as follows %$N$-partite extension of Eq. (\ref{eq2}), given by

\begin{eqnarray}
\widetilde{\mathcal{K}}_{I_{0}I_{1}\cdots I_{d-1}} &=& \mathbb{I}_{\text{A}}\otimes\sum_{j=0}^{d-1}\Big(\bigotimes_{n=1}^{N}E^{(j)}_{i_j,n}   E^{( j \oplus 1)}_{i_{j\oplus1},n}   \cdots  E^{(j \oplus (d-1) )}_{i_{j\oplus  (d-1)},n} \Big) \otimes |j\>\<j| \nonumber \\
 &=& \mathbb{I}_{\text{A}}\otimes\sum_{j=0}^{d-1}\Big(\bigotimes_{n=1}^{N} s_{j,n} |j\>_{\text{B}_n}\<i_{j \ominus 1,n}|_{\text{A}_n} \Big) \otimes |j\>\<j|_{\text{C}},
\label{eq:Krausdisgen}
\end{eqnarray}
where $\forall n, ~s_{j,n} = \prod_{l \neq d \ominus 1} \<i_{l,n}|l \oplus 1 \>$.
%where, $\mathcal{P}_{j}(E^{(0)}_{i_{0,n}}E^{(1)}_{i_{1,n}}\cdots E^{(d-1)}_{i_{d-1,n}})$ denotes the $j^{th}$-cyclic permutation on the product of $d$ Kraus operators $E^{(0)}_{i_{0,n}}E^{(1)}_{i_{1,n}}\cdots E^{(d-1)}_{i_{d-1,n}}$ along the $n^{th}$-transmission line. 
Again, one can identify $E_{i_{p,n}}^{(q)}=\ketbra{q}{i_{p}}, \forall p,q \in \{0,1,\cdots, (d-1)\}$ as the $i_p^{th}$ Kraus operator of the qudit information-erasing channel $\mathcal{E}_{q}$ acting on the $n^{th}$ transmission line.
Now $s_{j,n}$ is non-zero only when either one of the following two cases occur\\
1. For all $n \in {1,2, \ldots N}$, $i_{l,n} = l \oplus 1$ for all $j \in {0,1,2, \ldots d-1}$,\\
2. For some $n$ values, say $k$ of them $(1 \leq k \leq N)$, $i_{l,n} = l \oplus 1$ holds for all $j$ except one. For the remaining $N - k$ cases, we have $i_{l,n} = l \oplus 1$ for all $j \in {0,1,2, \ldots d-1}$.  \\

The Kraus operator corresponding to case 1. is unique and is given by $\mathbb{I}\otimes\sum_{j = 0}^{d-1} |j\>\<j|^{\otimes N}\otimes |j\>\<j| =: P_0^N$. The structure of the Kraus operators for case 2. is given by
\begin{eqnarray}
\mathbb{I}\otimes|j\>^{\otimes N} \<x_k(p)|\otimes \ketbra{j}{j},
\label{eq:diskrauscase2order}
\end{eqnarray}
where $|x_k(p)\>$ is a ket with $N$ elements such that $k$ of them are different from $j$. Clearly the elements of $|x_k(p)\>$ can be chosen in $^NC_{k} = \frac{N!}{k!(N-k)!}$ ways, and $p$ denotes the $p^{\text{th}}$ configuration. For example for $N = 2$, $k$ can take two values, namely $1$ and $2$. Now for $k = 1$, we have $|x_1(1)\> = |jl\>$ and $|x_1(2)\> = |lj\>$.  For $k = 2$, we get just $|x_2(1)\> = |l_1 l_2\>$.

\begin{lemma}\label{l10}
The controlled operation of $N$-noisy channels each consisting of $d$ orthogonal information-erasing channels using a perfect side channel for the control qudit maps every $(N+1)$ qudit input state to the subspace spanned by $\{\ket{0}^{\otimes(N+1)},\ket{1}^{\otimes(N+1)},\cdots,\ket{d-1}^{\otimes(N+1)}\}$. This also constitutes the decoherence free subspace.
\end{lemma}
\begin{proof}
From the expression of the Kraus operators arising from cases 1 and 2, it is clear that the $(N + 1 )$ qudit $A_1A_2\cdots A_NC$ subsystem of any arbitrary state $\rho_{\text{AA}_1\text{A}_{2}\cdots\text{A}_{N}\text{C}}$ would be mapped into the subspace spanned by $\{|j\>^{\otimes (N+1)} \}_{j=0}^{d-1}$. Furthermore, using these expressions, we can rewrite Eq. \eqref{eq:Krausdisgen} as 
\begin{align}\label{e36}
    \mathbb{I}_{\text{A}}\otimes\mathcal{S}^{(N)}(\tilde{\mathcal{E}}_{0},\cdots,\tilde{\mathcal{E}}_{d-1})[\rho_{\text{AA}_1\text{A}_2 \ldots \text{A}_N \text{C}}]   &=\nonumber \\
    P_0^N  \rho_{\text{AA}_1\text{A}_2 \ldots \text{A}_N \text{C}}  P_0^N +  \sum_{k = 1}^N \sum_{j=0}^{d-1} & \sum_{l_1,l_2 \ldots l_k \not  =  j} \sum_{p=1}^{^NC_k}  [\mathbb{I}_{\text{A}} \otimes \<  x_k^{\{l_i\}}(p)|\< j|  (\rho_{\text{AA}_1\text{A}_2 \ldots \text{A}_N \text{C}}) \mathbb{I}_{\text{A}} \otimes  |x_k^{\{l_i\}}(p)\>  |j\>] ~  |j\>\<j|^{\otimes N}  \otimes |j\>\<j|  \, . 
\end{align}
Here the additional superscript $\{l_i \}$ in $|x_k^{\{l_i\}}(p)\>$ denotes the values of the $k$ elements that are different from $j$.
%Note that the second term in Eq. \eqref{eq:Krausdisgen} accounts for the decoherence induced by the channel configuration. 

The same channel is obtained from a controlled choice of the information-erasing channels $\{\map E_{j}\}_{j=0}^{d-1},$ for each of the $N$ noisy transmission lines when, as before,  one considers the extended channels $\{\widetilde{\map E_{j}}\}_{j=0}^{d-1}$ with Kraus operators
$\widetilde E^{(j)}_{i_j}   = |j\>\<i_j|  +    \<i_j| j\> \,  |{\rm triv}\>\<{\rm triv}|$.  The controlled-choice channel is then given by \cite{abbott2020communication, chiribella2019quantum} 
\begin{equation}\label{eq:dischoicekraus}
    \mathbb{I}_{\text{A}}\otimes\mathcal{T}(\widetilde{\mathcal{E}}_{0}^{\otimes N}, \widetilde{\mathcal{E}}_{1}^{\otimes N}, \cdots, \widetilde{\mathcal{E}}_{d-1}^{\otimes N})(\rho_{\text{AA}_1\text{A}_2 \ldots \text{A}_N \text{C}})=\sum_{I_{0},I_{1},\cdots,I_{d-1}}T_{I_{0},I_{1},\cdots,I_{d-1}}\rho_{\text{AA}_1\text{A}_2 \ldots \text{A}_N \text{C}}T_{I_{0},I_{1},\cdots,I_{d-1}}^{\dagger}
\end{equation}
with  Kraus operators
\begin{equation}\label{eq:dischoice}
    T_{I_{0},I_{1},\cdots,I_{d-1}}=\mathbb{I}_{\text{A}}\otimes\sum_{j=0}^{d-1} \Big(\bigotimes_{n=1}^{N}   t_{j,n} \,   |j\>_{\text{B}_{n}}\<i_{j,n}|_{\text{A}_{n}} \Big)   \otimes |j\>\<j|_{\text{C}}   \, ,
    \qquad \qquad   t_{j,n}:  =   \prod_{l\not  = j}   \alpha_{i_{l,n}}^{(l)}    \, ,
\end{equation}
where $\alpha^{(l)}_{i_{l,n}}$ are the amplitudes associated to the $l$-th channel of the $n$-th transmission line.  If we set $\alpha^{(l)}_{i_{l,n}}  =  \<i_{l,n}|  l\>$, then there are two  possible cases for which the Kraus operators become non-zero:
\begin{enumerate}
\item $i_{l,n}=l$ for every $n\in\{1,2,\cdots,N\}$ and for every $l\in\{0,1,\cdots,d-1\}$,
\item For some $n$ values, say $k$ of them $(1 \leq k \leq N)$, $i_{l,n} = l$ holds for all except one value of $l$. For the remaining $N-k$ cases, $i_{l,n}=l$ for all $l\in\{0,1,\cdots,d-1\}$.
%\item \textcolor{red}{$i_{l,n}\neq l$ for more than one value of $l$ and for any numbers of $n\in\{1,2,\cdots,N\}$.}
\end{enumerate}
For Case 1, the Kraus operator is $T_{\bf{0},\bf{1},\dots,  \bf{d-1}}  =  P_0^N$, where $\bf{k}$ is an $n$ tuple with all elements equal to $k$.  The Kraus operators for Case 2,  is of the form
\begin{eqnarray}
\mathbb{I}\otimes|l\>^{\otimes N}\<x_k(p)|\otimes \ketbra{l}{l}
\end{eqnarray}
%$|l\>\<i_l| \otimes |l\>\<l|$, where $l$ is the one value such that $i_l \not =  l $.
Here $|x_k(p)\>$ is a ket with $N$ elements such that $k$ of them are different from $l$. As also mentioned after Eq. \eqref{eq:diskrauscase2order}, $p$ denotes the $p^{\text{th}}$ configuration out of a total of $^NC_k$ configurations.
We now arrive at the following relation  
\begin{align}\label{edisS=C}
    \mathbb{I}\otimes\mathcal{S}^{(N)}(\mathcal{E}_{0}^{\otimes N}, \mathcal{E}_{1}^{\otimes N}, \cdots, \mathcal{E}_{d-1}^{\otimes N})   =    \mathbb{I}\otimes\mathcal{T}^{(N)}(\widetilde{\mathcal{E}}_{0}^{\otimes N}, \widetilde{\mathcal{E}}_{1}^{\otimes N}, \cdots, \widetilde{\mathcal{E}}_{d-1}^{\otimes N})  =:  \map K^{(N)}  \, .
\end{align}
It generalizes Eq. (\ref{e29}) for arbitrary number of noisy transmission lines.  In the following we will treat the controlled-order and controlled-choice in a unified way, referring to the channel $\map K^{(N)}$.  

Now if $\rho_{\text{AA}_1\text{A}_2 \ldots \text{A}_N \text{C}}$ is so chosen such that the span of its $\text{A}_{1}\text{A}_{2}\cdots\text{A}_{N}\text{C}$ subsystem is $\{|j\>^{\otimes (N+1)} \}_{j=0}^{k}$, where $k \leq d-1$, the second term of Eq. \eqref{e36} vanishes identically. Hence its evolution is controlled by $P_0^N$ which in turn keeps it unchanged. Since this is true for any  $\rho_{\text{A}_1\text{A}_2 \ldots \text{A}_N \text{C}}$ in the subspace $\{|j\>^{\otimes (N+1)} \}_{j=0}^{k}$ $(k \leq d-1)$, we conclude that $\{|j\>^{\otimes (N+1)} \}_{j=0}^{d-1}$ is a decoherence free subspace. 
\end{proof}

%Before presenting the main proof of the \textit{third} theorem of the main text, let us state the following lemma which is directly follows from Eq. \eqref{e36}.

\subsection{The \textit{if} part of Theorem. 3}

Consider initially a maximally entangled state $\frac{1}{\sqrt{d}}\sum_{i=0}^{d-1}  \ket{i_{\text{A}}i_{\text{C}}}$ be shared between Alice and Charlie. 
Alice then makes local operations on her lab with additional ancillary qubits to extend her state to  an $N+2$ qudit GHZ state   $\frac{1}{\sqrt{d}}\sum_{i = 0}^{d-1}|j\>^{\otimes (N+2)}$ shared between A, A$_1$, A$_2 \ldots$ A$_N$, and C. Therefore we have $\rho_{\text{AA}_1\text{A}_2 \ldots \text{A}_N \text{C}} = |GHZ\>_{N+2}$.
Now using Eq. \eqref{e36} and lemma \ref{l10}, we have $ \mathcal{K}^{(N)}(\tilde{\mathcal{E}}_{0},\cdots,\tilde{\mathcal{E}}_{d-1})(|GHZ\>_{N+2}) = |GHZ\>_{N+2}$.  Therefore, the final state shared between Alice (A), the $N$ spatially separated Bobs,  (B$_1$, B$_2 \ldots$ B$_N$) and Charlie (C) is $|GHZ\>_{N+2}$. Now Charlie performs a measurement on his qudit in a $d$-dimensional Fourier basis $\{\ket{f_{m}}=\frac{1}{\sqrt{d}}\sum_{j=0}^{d-1}\exp(i\frac{2\pi j m}{d})\ket{j}\}_{m=0}^{d-1}$ and communicates the measurement outcome to the Bobs. They can now apply local unitaries to share an $(N+1)$ qudit $|GHZ\>_{N+1}$ among Alice and themselves. 
This completes the \textit{if} part of the proof.

%$\rho_{\text{A}_1\text{A}_2 \ldots \text{A}_N \text{C}}$ to be an $N$ qudit GHZ state $\frac{1}{\sqrt{d}}\sum_{i = 0}^{d-1}|j\>^{\otimes N}$.

\subsection{The \textit{only if} part of Theorem. 3}
Here we will prove the necessity of shared maximally entangled state between Alice and Charlie, with the help of the following Lemmas.
\begin{lemma}\label{l11}
The state shared between Alice, all the $N$ Bobs and Charlie, after the controlled quantum operation, must be $(N+2)$ qudit  GHZ state.
\end{lemma}
\begin{proof}
Note that, the Lemma can be seen as a generalization of Lemma \ref{l7} and here we will use the same flow of arguments.\par
Suppose the state shared between Alice, $N$ Bobs and Charlie after controlled quantum operation is $\rho_{\text{AB}_{1}\text{B}_{2}\cdots\text{B}_{N}\text{C}}$. Performing a local measurement on his qudit, Charlie can communicate the result to the Bobs, who then able to apply proper local operations to obtain a $(N+1)$-qudit GHZ state among Alice and themselves.\par 
Further, it is possible to distill a two-qudit maximally entangled between Alice and exactly one Bob, if all $N$ Bobs are allowed to perform only LOCC. Since, LOCC can not increase entanglement, this implies that $\rho_{\text{AB}_{1}\text{B}_{2}\cdots\text{B}_{N}\text{C}}$ will be maximally entangled in $A|B_{1}B_{2}\cdots B_{N}C$ bipartition. Hence, the marginal of the subsystem $A$ and atleast one among $N$-Bobs and Charlie should be maximally mixed $\frac{I}{d}$. This, along with Lemma \ref{l10}, helps us to conclude that $\rho_{\text{AB}_{1}\text{B}_{2}\cdots\text{B}_{N}\text{C}}$ is pure and can be written in Schmidt form in the $A|B_{1}B_{2}\cdots B_{N}C$ bipartition,
\begin{equation}\label{e41}
     \ket{\psi^{(N)}}_{\text{AB}_{1}\text{B}_{2}\cdots\text{B}_{N}\text{C}}=\frac{1}{\sqrt{d}}\sum_{j=0}^{d-1}e^{i \theta_{j}}\ket{\psi_{j}}_{\text{A}}\otimes\ket{j}^{\otimes (N+1)}_{\text{B}_{1}\text{B}_{2}\cdots\text{B}_{N}\text{C}}.
\end{equation}
where, $\{\ket{\psi_{j}}\}_{\text{A}}$ is any arbitrary orthogonal basis for Alice's subsystem. One can identify the state Eq. (\ref{e41}) as a generalization of Eq. (\ref{e30}). 

This completes the proof.
\end{proof}
%Now, the state $\rho^{*}_{\text{AA}_{1}\text{A}_{2}\cdots\text{A}_{N}C}$, used as input in the controlled quantum operation, is locally prepared by Alice from a state $\rho_{\text{AC}}$ initially shared with Charlie. 
This helps us to conclude the following Corollary (generalization of Corollary \ref{coro1}) regarding the entanglement content of the state $\rho^{*}_{\text{AA}_{1}\text{A}_{2}\cdots\text{A}_{N}C}$ used as an input for the controlled quantum operations.\par
\begin{corollary}\label{coro2}
The state $\rho^{*}_{\text{AA}_{1}\text{A}_{2}\cdots\text{A}_{N}C}$ should be maximally entangled in $A|A_{1}A_{2}\cdots A_{N}C$ bipartition.
\end{corollary}
\begin{proof}
Similar to Corollary \ref{coro1}, the proof follows from the fact that operation performed on the $A_{1}A_{2}\cdots A_{N}C$ subsystem of the state $\rho^{*}_{\text{AA}_{1}\text{A}_{2}\cdots\text{A}_{N}C}$ can not increase the entanglement content in the $A|B_{1}B_{2}\cdots B_{N}C$ bipartition of the final state. Now, thanks to Eq. (\ref{e41}) of Lemma \ref{l11}, the final state with $\log d$-bit of entanglement completes the proof. 
\end{proof}
At this end, we will now prove the \textit{only if} part of Theorem. 3, both for the controlled-order and controlled-choice configuration, separately.

\medskip

\textbf{Proof for the controlled-order configuration.}

\medskip

\begin{lemma}\label{l12}
To establish a perfect GHZ state among Alice and all $N$ Bobs after the controlled-order configuration of $d$ orthogonal pin-maps in each of $N$ transmission lines, the state shared between Alice and Charlie must be maximally entangled.
\end{lemma}
\begin{proof}
The proof follows as a generalization of Lemma \ref{l8}.\par
Note that, we can conclude from Corollary \ref{coro2} that the state $\rho^{*}_{\text{AA}_{1}\text{A}_{2}\cdots\text{A}_{N}\text{C}}$ is maximally entangled in $A|B_{1}B_{2}\cdots B_{N}C$ bipartition. Which can only be preserved after the controlled-order operation only if the separable contribution for each transmission lines from Eq. (\ref{e36}) vanish and hence

$$\sum_{l_1,l_2,\cdots,l_k\neq j}(\mathbb{I}_{\text{A}}\otimes\bra{x_{k}^{\{l_i\}}(p)}\bra{j})\rho_{\text{AA}_{1}\cdots\text{A}_N\text{c}}(\mathbb{I}_{\text{A}}\otimes\ket{x_{k}^{\{l_i\}}(p)}\ket{j})=0$$
for every choice of $^NC_{k}$ possibilities, for each $k\in\{1, 2, \cdots, N\}$ and for every $j\in\{0, 1, \cdots, d-1\}$. This, in turn, claims that the state $\sigma^{*}_{\text{A}_{1}\cdots\text{A}_{N}\text{C}}:=\text{Tr}_{\text{A}}[\rho^{*}_{\text{AA}_{1}\cdots\text{A}_{N}\text{C}}]$ is orthogonal to the subspace $\mathcal{S}_{\perp}$, where $\mathcal{S}_{\perp}$ is the subspace orthogonal to $\ket{j}^{\otimes N}\otimes\ket{j}, ~\forall j\in \{0, 1, \cdots, d-1\}$.\par
This, along with the result of Corollary \ref{coro2}, demands that the state $\rho^{*}_{\text{AA}_{1}\cdots\text{A}_{N}\text{C}}$ can be expressed as,
$$\ket{\psi^{*}}_{\text{AA}_1\cdots\text{A}_N\text{C}}=\frac{1}{\sqrt{d}}\sum_{j=0}^{d-1}e^{i\phi_{j}}\ket{\psi_{j}}_{\text{A}}\otimes\ket{j}_{\text{A}_1\cdots\text{A}_N}^{\otimes N}\otimes\ket{j}_ \text{C},~\text{where } \<\psi_{k}\ket{\psi_{l}}=0, ~\forall k\neq l.$$
Now, performing consecutive Fourier measurement on all $N$ qudits $\{A_1,A_2,\cdots,A_N\}$ and performing local operations Alice can generate maximal entanglement between Charlie and herself, which once again confirms that the initial state shared between Alice and Charlie must have $\log d$-bits of entanglement.
\end{proof}

\medskip

\textbf{Proof for the controlled-choice configuration.}\par
Let us first consider the generalization of Eq. (\ref{e31}), for $N$ noisy transmissions lines, each with $d$ orthogonal pin-maps, in a controlled-choice configuration with the help of a perfect qudit side-channel. This reads, 
\begin{eqnarray}\label{e42}
    \nonumber\mathbb{I}_{\text{A}}\otimes\mathcal{T}^{(N)}(\widetilde{\mathcal{E}}^{\otimes N}_{0},\widetilde{\mathcal{E}}^{\otimes N}_{1},\cdots,\widetilde{\mathcal{E}}^{(N)}_{d-1})[\rho_{\text{AA}_1\text{A}_2\cdots\text{A}_N\text{C}}]=\mathbb{I}_{\text{A}}\otimes T^{(N)}_{\textbf{0},\textbf{0},\cdots,\textbf{0}}~\rho_{\text{AA}_1\text{A}_2\cdots\text{A}_N\text{C}}~\mathbb{I}_{\text{A}}\otimes T^{(N)\dagger}_{\textbf{0},\textbf{0},\cdots,\textbf{0}}+\\
    \sum_{k=1}^{N}\sum_{p_k=1}^{^NC_k}\sum_{j=0}^{d-1}\text{Tr}_{\text{A}_1\cdots\text{A}_N\text{C}}[\mathbb{I}_{\text{A}}\bigotimes_{n\in p_k}(\mathbb{I}_{A_n}-\ketbra{v_j^n}{v_j^n})\bigotimes_{m\notin p_k}\ketbra{j}{v_j^m}\otimes\ketbra{j}{j}_{\text{C}}~\rho_{\text{AA}_1\text{A}_2\cdots\text{A}_N\text{C}}]\ketbra{j}{j}^{\otimes N}_{\text{B}_1\cdots\text{B}_N}\otimes\ketbra{j}{j}_{\text{C}}
\end{eqnarray}
where, $p_k$ is all possible choice of $k\in\{1,2,\cdots,N\}$ transmissions lines among $N$. We will now conclude the result finally with the following Lemma.\par
 \begin{lemma}
 It is possible to obtain the state $\ket{\psi^{(N)}}_{\text{AB}_1\text{B}_2\cdots\text{B}_N\text{C}}$ of Eq. (\ref{e41}) after controlled-choice of $N$ transmission lines, with $d$ orthogonal pin-maps in each, only if Alice and Charlie shares a maximally entangled state initially.
 \end{lemma}
\begin{proof}
To preserve the maximal entanglement in the $A|A_{1}A_{2}\cdots A_{N}C$ bipartition of the state $\rho^{*}_{\text{AA}_1\text{A}_2\cdots\text{A}_N\text{C}}$, every possible separble terms in Eq. (\ref{e42}) must vanish. In a similar argument of that of Lemma \ref{l9}, the condition simply implies
$$T^{(N)}_{\textbf{0},\textbf{0},\cdots,\textbf{0}}=\mathbb{I}_{\text{A}}\otimes\sum_{j=0}^{d-1}\bigotimes_{n=1}^{N}\ketbra{j}{v_{j,n}}\otimes\ketbra{j}{j}_{\text{C}}$$
where, for every $j\in\{0, 1, \cdots, d-1\}$ and for every $n\in\{1, 2, \cdots, N\}, ~ ||v_{j,n}||=1$. This, along with Corollary \ref{coro2}, demands a that the state $\rho^{*}_{\text{AA}_1\text{A}_2\cdots\text{A}_N\text{C}}$ is pure and can be written in the form
$$\ket{\phi^{(N)*}}_{\text{AA}_1\cdots\text{A}_N\text{C}}=\frac{1}{\sqrt{d}}\sum_{j=0}^{d-1}e^{i\theta_j}\ket{\psi_{j}}_{\text{A}}\bigotimes_{k=1}^{N}\ket{v_{j,n}}_{\text{A}_k}\otimes\ket{j}_{\text{C}}$$
where, $\{\ket{\psi_{j}}\}$ is a orthonormal basis for the subsystem $A$ and the state can be identified as the generalization of Eq. (\ref{e32}). 

Now, Alice is able to apply the joint unitary $U_{\text{AA}_1\cdots\text{A}_N}$ on the $AA_1A_2\cdots A_N$ subsystem which takes 
$$U_{\text{AA}_1\cdots\text{A}_N}\ket{\psi_j}_{\text{A}}\bigotimes_{k=1}^{N}\ket{v_{j,n}}_{\text{A}_k}\to\ket{\psi_j}_{\text{A}}\bigotimes_{k=1}^{N}\ket{w_{j,n}}_{\text{A}_k}$$
and hence
$$(U_{\text{AA}_1\cdots\text{A}_N}\otimes\mathbb{I}_{\text{C}})\ket{\phi^{(N)*}}_{\text{AA}_1\cdots\text{A}_N\text{C}}\to\ket{\xi^{(N)*}}_{\text{AA}_1\cdots\text{A}_N\text{C}}:=\frac{1}{\sqrt{d}}\sum_{j=0}^{d-1}e^{i\theta_j}\ket{\psi_{j}}_{\text{A}}\bigotimes_{k=1}^{N}\ket{w_{j,n}}_{\text{A}_k}\otimes\ket{j}_{\text{C}}$$
where $\{\ket{w_{j,n}}\}_{j=0}^{d-1}$ is an orthogonal basis for every $n\in\{1,2,\cdots,N\}$.

Now just performing Fourier basis measurement on every $A_k,~k\in\{1,2,\cdots,N\}$ and performing suitably chosen unitary on $A$ subsystem, Alice can establish $\log d$-bit of entanglement with Charlie. This, in turn, assures that the state initially shared between Alice and Charlie should contain $\log d$-bit of entanglement, i.e., a maximally entangled two-qudit state. 
\end{proof}
\end{widetext}
%{\color{blue}  [Add your names in the list of authors, add weblink for all references, add reference to supplemental material]}
\end{document}